\documentclass[final,journal,letterpaper]{IEEEtran}

\usepackage{times}
\usepackage{amsmath,amssymb,amsfonts,amsthm}
\usepackage{algorithm}
\usepackage{algpseudocode}
\usepackage{dsfont}
\usepackage{color}
\usepackage{graphicx}
\usepackage{subfigure}
\usepackage{enumitem}
\usepackage{comment}
\usepackage{cite}

\algnewcommand{\LeftComment}[1]{ \(\triangleright\) #1}
\newcommand{\mypara}[1]{{\smallskip \noindent \bf #1}\hspace{0.1in}}

\newtheorem{theorem}{Theorem}
\newtheorem{lemma}{Lemma}
\newtheorem{definition}{Definition}

\newtheorem{corollary}{Corollary}
\allowdisplaybreaks

\begin{document}

\title{On No-Sensing Adversarial Multi-player Multi-armed Bandits with Collision Communications}

\author{Chengshuai~Shi,~\IEEEmembership{Student Member,~IEEE} and Cong~Shen,~\IEEEmembership{Senior Member,~IEEE}
    \thanks{The work was supported in part by the US National Science Foundation (NSF) under Grant CNS-2002902 and ECCS-2029978, and a Commonwealth Cyber Initiative (CCI) cybersecurity research collaboration grant.}
	\thanks{The authors are with Charles L. Brown Department of Electrical and Computer Engineering, University of Virginia, Charlottesville, VA 22904, USA. Email: \texttt{\{cs7ync, cong\}@virginia.edu}.}
	}

	\maketitle

	\begin{abstract}
    We study the notoriously difficult no-sensing adversarial multi-player multi-armed bandits (MP-MAB) problem from a new perspective. Instead of focusing on the hardness of multiple players, we introduce a new dimension of hardness, called \emph{attackability}. All adversaries can be categorized based on the attackability and we introduce \emph{Adversary-Adaptive Collision-Communication (A2C2)}, a family of algorithms with forced-collision communication among players. Both attackability-aware and unaware settings are studied, and information-theoretic tools of the Z-channel model and error-correction coding are utilized to address the challenge of implicit communication without collision information in an adversarial environment. For the more challenging attackability-unaware problem, we propose a simple method to estimate the attackability enabled by a novel error-detection repetition code and randomized communication for synchronization. Theoretical analysis proves that asymptotic attackability-dependent sublinear regret can be achieved, with or without knowing the attackability. In particular, the asymptotic regret does not have an exponential dependence on the number of players, revealing a fundamental tradeoff between the two dimensions of hardness in this problem.

	\end{abstract}

	\section{Introduction}\label{sec:intro}
	The decentralized multi-player multi-armed bandits (MP-MAB) problem has received increasing interest in recent years  \cite{liu2010distributed,besson2018multi,boursier2019sic,lugosi2018multiplayer,bubeck2020non}. In MP-MAB, multiple players simultaneously play the bandit game without explicit communications and interact with each other only through arm collisions. When two or more players play the same arm simultaneously, they all get a reward $0$ (or equivalently loss $1$) instead of the true underlying reward of that action. This model is largely motivated by practical applications such as cognitive radio \cite{anandkumar2011distributed,avner2014concurrent,Gan2020tsp,Shi2020aistats} and wireless caching \cite{Xu2020twc}, where standard (single-player) MAB does not fully capture the system complexity and user interactions must be taken into account in conjunction with the bandit game.

	Depending on how rewards are generated, the MP-MAB game can be either stochastic or adversarial, as in the single-player bandit problem. Most of the existing MP-MAB works focus on the stochastic setting, in which a well-behaved stochastic model exists for each arm (albeit unknown to the players). However, it is often difficult to determine the correct stochastic assumptions in real-world applications, and there are use cases where such assumptions do not hold. For example, in cognitive radio systems, it is common to have channel availability or signal quality fluctuations due to changing environmental conditions or bursty radio frequency interference \cite{Shen2019tsp,alatur2020multi}. Adversarial MP-MAB is a more suitable model for such use cases, as it makes no stochastic assumptions on the rewards and assigns an arbitrary reward sequence to each arm exogenously. However, compared with stochastic MP-MAB, adversarial MP-MAB is a considerably harder problem because of the need to fight the adversary while interacting with other players.

	Since the MAB problem for a single player is well understood, a predominant approach for both stochastic and adversarial MP-MAB is to let each player play the single-player MAB game while avoiding collisions for as much as possible \cite{liu2010distributed,avner2014concurrent,rosenski2016multi,besson2018multi,bande2019adversarial}.
	Recently, a pioneering work \cite{boursier2019sic} proposes to purposely instigate collisions as a way to share information between players. Such \emph{implicit communication} is instrumental in breaking the performance barrier and achieving a regret that approaches the centralized multi-play MAB \cite{anantharam1987asymptotically,komiyama2015optimal}. This idea has been extended to several variants in the stochastic setting \cite{proutiere2019optimal,tibrewal2019multiplayer,magesh2019multi,boursier2019practical} as well as adversarial MP-MAB \cite{alatur2020multi,bubeck2020non}, with improved regret performance for all models.

	All the aforementioned works make an important assumption of \emph{collision sensing} -- any collision with another player is perfectly known. Such ``collision indicator'' plays a fundamental role in both collision avoidance and forced-collision communication. It is widely recognized that a more difficult problem in MP-MAB is the \emph{no-sensing} scenario, in which players can only observe the final rewards but not collisions. The difficulty lies in that the zero rewards can indistinguishably come from collisions or null arm rewards. Recently, there is some progress on the stochastic no-sensing problem \cite{lugosi2018multiplayer,bubeck2020coordination}. In particular, the fundamental idea of {\em implicit communication} is again proved crucial in achieving regret that approaches the centralized counterpart \cite{Shi2020aistats}.

	Nevertheless, the most difficult setting of no-sensing adversarial MP-MAB in a fully decentralized setting remains wide open. To the best of the authors' knowledge, reference \cite{bubeck2020non} is the only work that achieves a sublinear regret by a collision-avoidance design (i.e., no implicit communication) where ``safe'' arms are reserved for players. However, its asymptotic regret of $O(T^{1-\frac{1}{2M}})$ is almost linear in $T$ when $M$ is large, where $M$ is the number of players and $T$ is the time horizon of the game. We note that this exponential dependence on $M$ reveals a particular {\em dimension of hardness} (multiple players) in the no-sensing adversarial MP-MAB problem.

	Recent development has repeatedly demonstrated that implicit communication is crucial in achieving lower regret. However, as pointed out in \cite{bubeck2020non}, it is unclear how to {implicitly communicate without collision information in an adversarial environment}. This paper makes progress in the no-sensing adversarial MP-MAB problem by addressing the challenges in incorporating implicit communication. In particular, this work reveals a novel dimension of hardness associated with the no-sensing adversarial MP-MAB problem: \textit{attackability} of the adversary, that is orthogonal to the multi-player dimension of hardness. Technically, we depart from the approach of \cite{bubeck2020non} which always assumes the worst possible adversary while focusing on the multi-player hardness, and analyze the relationship between the attackability hardness and implicit communications. Notably, \emph{all} possible adversaries can be classified based on this new concept of attackability, which is defined either by a local view (for a ``one-time'' attack) or a global view (for the cumulative attacks). The hardness of attackability may or may not be aware by the players, and we develop a suite of \emph{Adversary-Adaptive Collision-Communication (A2C2)} algorithms under both attackability-aware and attackability-unaware settings, which adaptively adjust the implicit communication by learning the attackability of the adversary in an online manner. All of the A2C2 algorithms utilize some (common) new elements that have not been considered before in no-sensing adversarial MP-MAB, such as an information-theoretical Z-channel model and error-correction coding, to design a forced-collision communication protocol that can effectively fight against the adversary and achieve a non-dominant communication regret in the no-sensing setting. On the other hand, for the more challenging attackability-unaware setting, we show that a simple ``escalation'' estimation of the attackability, a novel error-detection repetition code, and randomized synchronizations are crucial to handle the unknown attackability. A key idea behind algorithms in the attackability-unaware setting is that \textit{communication error is not bad if it happens to {all} players}, as such error does not affect player synchronization.

	The regret analysis of the A2C2 algorithms shows that they can achieve attackability-dependent sublinear regrets asymptotically, without an exponential dependence on the number of players as in \cite{bubeck2020non}. This benefit, however, does not lead to a universally lower regret. In fact, we may view A2C2 of this paper and the method of \cite{bubeck2020non} as operating at two different regimes in the two-dimensional hardness space (multi-player and attackability). On one hand, the $T$ terms in the regret of A2C2 algorithms are oblivious to the number of players (i.e., $M$) and the overall dependency on $M$ is only a multiplicative factor, while the regret of \cite{bubeck2020non} has an exponential dependency on $M$. On the other hand, A2C2 algorithms have exponential dependencies on the attackability, which does not affect the regret of \cite{bubeck2020non}. Philosophically speaking, the regret comparison between A2C2 and \cite{bubeck2020non} shows that {\em one can trade off the multi-player dimension of hardness with the attackability dimension of hardness}, which may provide insight into other relevant adversarial bandit problems. A comparison of the regret bounds are given in Table \ref{table:regret} for both collision-sensing and no-sensing adversarial MP-MAB algorithms.

	The rest of the paper is organized as follows. Related works are surveyed in Section~\ref{sec:related}. The no-sensing adversarial MP-MAB problem is formulated in Section~\ref{sec:problem}. The general algorithm structure is presented in Section~\ref{sec:out}, followed by algorithms for known (Section~\ref{sec:kno}) and unknown (Section~\ref{sec:unkno}) attackability. The regret analyses of all algorithms are given in Section~\ref{sec:theory}. Discussions on some algorithmic details and future research directions are given in Section~\ref{sec:dis}. Numerical illustrations and experimental results are provided in Section~\ref{sec:exp} to support the theoretical analyses. Finally, Section~\ref{sec:conc} concludes the paper.

	\begin{table*}[htb]
		\caption{Regret Bounds of Adversarial MP-MAB Algorithms }
		\begin{center}
			\begin{tabular}{|lll|}
				\hline
				\textbf{Model}  &\textbf{Reference} &\textbf{Asymptotic Bound} \\
				\hline
				Centralized Multiplayer, Optimal Regret &\cite{audibert2013regret} &$\Theta\left(\sqrt{MKT}\right)$\\
				Decentralized, Collision Sensing &\cite{bande2019adversarial} &$\tilde{O}\left(K^{M}M^{2}T^{\frac{3}{4}}\right)$ \\
				Decentralized, Collision Sensing &\cite{alatur2020multi} &$\tilde{O}\left(M^{\frac{4}{3}}K^{\frac{1}{3}}T^{\frac{2}{3}}\right)$ \\
				Decentralized, Collision Sensing, $M=2$ &\cite{bubeck2020non} &$\tilde{O}\left(K^2\sqrt{T}\right)$\\
				Decentralized, No Sensing &\cite{bubeck2020non} &$\tilde{O}\left(MK^{\frac{3}{2}}T^{1-\frac{1}{2M}}\right)$\\
				Decentralized, No Sensing, Local Attackability & $\alpha$-aware A2C2 (this work) &$\tilde{O}\left(M^{\frac{4}{3}}K^{\frac{1}{3}}T^{\frac{2+\alpha+\epsilon}{3}}\right)$ \\
				Decentralized, No Sensing, Global Attackability & $\beta$-aware A2C2 (this work) &$\tilde{O}\left(M^2K^{\frac{2}{3}}T^{\max\{\frac{1+\beta}{2}, \frac{2}{3}\}}\right)$ \\
				Decentralized, No Sensing, Local Attackability & $\alpha$-unaware A2C2 (this work) &$\tilde{O}\left(M^{\frac{4}{3}}K^{\frac{1}{3}}T^{\frac{5+\alpha+\epsilon}{6}}\right)$ \\
				Decentralized, No Sensing, Global Attackability & $\beta$-unaware A2C2 (this work) &$\tilde{O}\left(M^2K^{\frac{1}{3}}T^{\max\{\frac{2+\beta+\epsilon}{3},\frac{3}{4}\}}\right)$\\
				\hline
			\end{tabular}
		\end{center}
		\begin{center}
			$K$: number of arms; $M$: number of players with $1<M\leq K$; $\epsilon$: an arbitrarily small constant;\\
			$\alpha$: local attackability (see Corollary \ref{asp:single_attack}); $\beta$: global attackability (see Corollary \ref{asp:overall_attack}).\\
			With the notation of $\tilde{O}(\cdot)$, the logarithmic factors of $T$ and $K$ are ignored.
		\end{center}
		\vspace{-0.2in}
		\label{table:regret}
	\end{table*}

 	\section{Related Work}
 	\label{sec:related}
 	\subsection{Overall Review}
 	\mypara{Collision-sensing stochastic and adversarial MP-MAB.} As stated in Section \ref{sec:intro}, initial approaches for collision-sensing MP-MAB adopt single-player MAB algorithms with various collision-avoidance protocols. Examples include Explore-then-Commit \cite{rosenski2016multi}, UCB \cite{liu2010distributed,besson2018multi}, $\epsilon$-greedy \cite{avner2014concurrent} for stochastic MP-MAB, and EXP3 \cite{bande2019adversarial} for adversarial MP-MAB with a regret of $O(T^\frac{3}{4})$. Although these strategies achieve sublinear regret, their performance cannot approach the centralized counterparts. In particular, for the stochastic environment, there is a multiplicative factor $M$ increase in the regret coefficient of $\log(T)$ compared with the natural lower bound of centralized MP-MAB \cite{anantharam1987asymptotically,komiyama2015optimal}, which has long been considered fundamental due to the lack of explicit communication among players.

 	The idea of implicit communication with forced collisions is introduced by the SIC-MMAB algorithm \cite{boursier2019sic}, where bits $1$ and $0$ are transmitted by collision and no collision, respectively. The theoretical analysis of SIC-MMAB shows, for the first time, that the regret of decentralized MP-MAB can approach the centralized lower bound in the stochastic environment.
 	The DPE1 algorithm \cite{proutiere2019optimal} further improves the regret by combining the KL-UCB algorithm \cite{garivier2011kl} with implicit communication. Similar ideas have also been extended to other stochastic variants, such as the heterogeneous setting \cite{tibrewal2019multiplayer,boursier2019practical}, where rewards are player-dependent. For the adversarial environment, implicit communication also proves to be effective. In particular, the C\&P algorithm \cite{alatur2020multi} achieves a regret of $O(T^{\frac{2}{3}})$  by invoking forced collisions to let players coordinately perform a centralized EXP3 algorithm. The performance for two players ($M=2$) has been improved to $O(\sqrt{T\log(T)})$ in \cite{bubeck2020non} by applying a filtering strategy with bandit-type information supported by implicit communication, which approaches the lower bound of $\Theta(\sqrt{T})$ \cite{audibert2013regret}.

 	\mypara{No-sensing stochastic and adversarial MP-MAB.} No-sensing MP-MAB represents a more challenging scenario and the progress has been limited. A collision-avoidance scheme is investigated in \cite{lugosi2018multiplayer} for the stochastic environment, which cannot approach the centralized lower bound.  Some initial attempts to incorporate implicit communication in the no-sensing stochastic setting, e.g., sharing arm indices instead of statistics, are discussed in \cite{boursier2019sic}. The EC-SIC algorithm proposed in \cite{Shi2020aistats} proves that it is possible to approach the centralized lower bound even without information of collision.  For the most difficult case of no-sensing adversarial MP-MAB, progress is extremely limited. To the best of our knowledge, \cite{bubeck2020non} is the only work studying this problem, and detailed comparisons between \cite{bubeck2020non} and A2C2 are provided in the next subsection.

 	\mypara{Cooperative MP-MAB.} This is another line of MP-MAB research where explicit communications are allowed (under certain constraints) and players do not collide with each other. Such scenarios have been studied in both stochastic and adversarial environments \cite{landgren2016distributed,shahrampour2017multi,wang2019distributed,awerbuch2008competitive,cesa2016delay}, which are under a completely different framework than this work.

 	\subsection{Comparison with \cite{bubeck2020non}}
 	As the only existing work studying the no-sensing adversarial MP-MAB problem, Ref. \cite{bubeck2020non} designs a novel collision-avoidance algorithm by reserving ``safe'' arms for players. Nevertheless, this approach gives up sharing information through implicit communication. As a result, intuitively, when there are more players involved in the game, it becomes increasingly difficult to avoid collisions without sufficient coordination enabled by (implicit) communications. This issue is reflected in its achievable regret of $O(T^{1-\frac{1}{2M}})$, which has an exponential dependency on $M$.  To address this critical issue of \cite{bubeck2020non}, this paper focuses on designing the implicit communication strategies through forced collisions and proposes the concept of attackability. The resulting suite of A2C2 algorithms have only multiplicative dependencies on $M$ in their achievable regrets. However, this is accomplished by incurring additional communication regrets caused by collisions, which has exponential dependencies on the attackability.
 	Thus, we conclude that A2C2 establishes an alternative dimension in the hardness space from attackability, in additional to the original dimension from multiple players in \cite{bubeck2020non}.
	\begin{figure}[htb]
		\centering
		\includegraphics[width=0.9\linewidth]{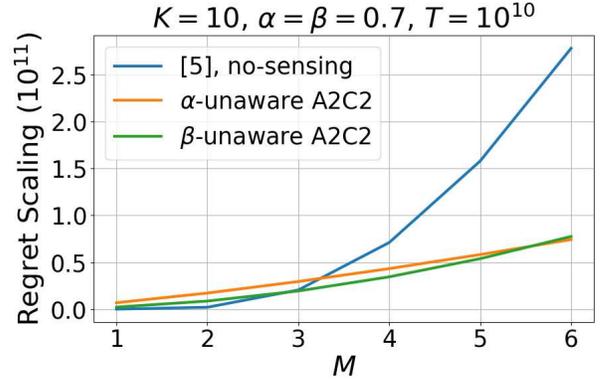}
		\caption{Multi-player dimension of hardness.}
		\label{fig:multiplayer}
	\end{figure}
	\begin{figure}[htb]
		\centering
		\includegraphics[width=0.9\linewidth]{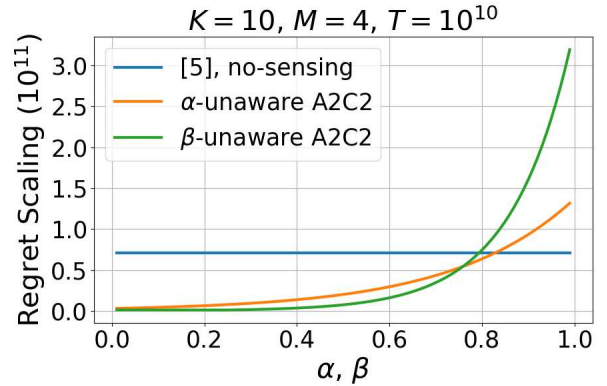}
		\caption{Attackability dimension of hardness.}
		\label{fig:attack}
	\end{figure}

    The following preview of the analytical results provides a more clear view of the two different dimensions of hardness. Fig.~\ref{fig:multiplayer} and Fig.~\ref{fig:attack} numerically illustrate the theoretical dependencies of the asymptotical regret (i.e. scaling) of A2C2 algorithms and the no-sensing algorithm in \cite{bubeck2020non} on the two dimensions of hardness. When fixing the attackability, i.e., fixing the local attackability parameter $\alpha$ to be $0.7$ or the global attackability parameter $\beta$ to be $0.7$, the regrets of $\alpha$-unaware A2C2 algorithm and the $\beta$-unaware A2C2 algorithm only increase slowly with $M$ in Fig.~\ref{fig:multiplayer}, since their $T$ terms are oblivious to $M$ and the overall dependency on $M$ is only a multiplicative factor. However, the regret of \cite{bubeck2020non} increases sharply with more players due to the exponential dependency. When $M$ is large (larger than $4$ in Fig.~\ref{fig:multiplayer}), the advantage of A2C2 algorithms is obvious. On the other hand, while fixing the number of players, the regret performance of \cite{bubeck2020non} is immune to the change of attackability. However, A2C2 algorithms have exponential dependencies on the attackability. As a result, their performances are very good when the adversary's attackability is weak or medium, but degrade quickly when the  attackability is extremely strong.

	\section{Problem Formulation}
	\label{sec:problem}

	\subsection{The no-sensing adversarial MP-MAB problem}
	\label{sec:model}

	We focus on the following decentralized no-sensing adversarial MP-MAB model. There are $K$ arms and $1<M\leq K$ players in the game\footnote{The special case of $M=1$ is addressed in Appendix \ref{app:singleplayer}.}. The arms are labeled  $1$ to $K$ and the players  $1$ to $M$, respectively. There are no \emph{explicit} communications among players, which is a key constraint of the decentralized MP-MAB problem. The time horizon $T$ is slotted and synchronized among players, and at each time step $t\in[T]$, each player $m\in[M]$ individually chooses and pulls arm $\pi_m(t)$. Simultaneously, an adversary selects loss $l_k(t)$ for each arm $k\in[K]$. The true loss $l_k(t)$ is player-independent and has a bounded support on $[0,1]$. A \emph{collision} happens if more than one player pull the same arm simultaneously. If no collision happens for player $m$ at time $t$, she receives the true loss $l_{\pi_m(t)}(t)$; otherwise, she always receives loss $1$ (i.e., the maximum loss) regardless of $l_{\pi_m(t)}(t)$. The actual loss $s_{\pi_m(t)}(t)$ received by player $m$ at time $t$ can be written as
	$$s_{\pi_m(t)}(t):=\underbrace{l_{\pi_m(t)}(t)(1-\eta_{\pi_m(t)}(t))}_{\text{no collision}}+\underbrace{\eta_{\pi_m(t)}(t)}_{\text{collision}},$$
	where $\eta_{\pi_m(t)}$ is the collision indicator defined as
	$\eta_{k}(t):=\mathds{1}\{|C_k(t)|>1\}$,
	with $C_k(t):=\{n\in[M]|\pi_n(t)=k\}$.

	If the players have access to both $s_{\pi_m(t)}(t)$ and $\eta_{\pi_m(t)}(t)$, it is a collision-sensing problem and player $m$ makes decision $\pi_m(t)$ with past information $\{\pi_m(v),s_{\pi_m(v)}(v),\eta_{\pi_m(v)}(v)\}_{v<t}$. When information of $\eta_{\pi_m(t)}(t)$ is unavailable and players only know $s_{\pi_m(t)}(t)$, the problem is a no-sensing one as considered in this paper. In this no-sensing setting, a loss $1$ can indistinguishably come from collisions or be exogenously generated by the adversary, and player $m$ makes decisions $\pi_m(t)$ only with $\{\pi_m(v), s_{\pi_m(v)}(v)\}_{v<t}$. The lack of information on the collision indicators complicates the MP-MAB problem in general \cite{boursier2019sic,Shi2020aistats}, and this challenge is more significant for the adversarial setting \cite{bubeck2020non}. Note that if $l_k(t)\neq 1$, $\forall k,t$, the no-sensing setting is equivalent to collision-sensing.

	In the adversarial MP-MAB model, the notion of regret can be generalized with respect to the best allocation of players to arms as follows \cite{alatur2020multi}:
	\begin{equation}\label{eqn:regret_def}
	    R(T):=\sum_{t=1}^T\sum_{m\in[M]}s_{\pi_m(t)}(t)-\min_{\begin{subarray}{c} k_1,...,k_M\in[K],\\ k_p\not=k_q,\forall p\not=q\end{subarray}}\sum_{t=1}^T\sum_{m\in[M]}l_{k_m}(t).
	\end{equation}
	We are interested in the \emph{expected regret} $\mathbb{E}[R(T)]$ where the expectation is with respect to the algorithm randomization.

	As shown in \cite{bubeck2020non}, one cannot obtain any non-trivial regret guarantees facing an adaptive adversary. This work thus focuses on the  \textit{oblivious} adversarial MP-MAB where the reward generation of the adversary is independent of the actions of players. Equivalently, the loss sequence is chosen by the adversary at the beginning of the game.

	\mypara{A motivating example.} The adversarial MP-MAB problem formulated in this section captures the key characteristics of a cognitive radio system, where arms correspond to channels, and players represent the distributed devices trying to communicate over the channels. The adversarial loss captures the unpredictable time-varying channel quality, e.g., due to bursty interference \cite{TV:05}. When two (or more) users use the same channel simultaneously, a packet collision happens and both of their attempted communications fail (i.e., a loss of $1$). This is a commonly used model for shared wireless medium, e.g., the CSMA protocol in Wi-Fi \cite{CSMA2012}.
	Regarding whether the collision is perceivable or not (i.e., collision-sensing or no-sensing), it is determined by the communication protocol and the sensing capabilities of devices. As stated in \cite{bonnefoi2017multi,lugosi2018multiplayer,besson2018multi}, the no-sensing setting is more suited to large scale Internet-of-Things (IoT) applications. Furthermore, minimizing the regret defined in Eqn.~\eqref{eqn:regret_def} is equivalent to minimizing the cumulative communication loss over the chosen channels, which is a meaningful metric for practical systems. As a final remark, since the changing of channel quality is typically caused by external factors such as the dynamic radio frequency environment instead of user actions, the assumption of an oblivious adversary is reasonable in the application of cognitive radio.

	\subsection{Attackabilities of the adversary}
	\label{sec:advattack}
	To explore the idea of forced-collision communication in the no-sensing adversarial setting, the overall horizon $T$ is divided into the exploration and communication phases, similar to the approaches in collision-sensing settings \cite{boursier2019sic, alatur2020multi}. Information is shared by purposely created collisions in the communication phases to maintain synchronization and coordination between players in the subsequent exploration phases. However, in the no-sensing setting, loss $1$ assigned by the adversary can be viewed as a certain ``attack'', since players have no knowledge whether it comes from the adversary or collision. Such loss-$1$ attack has very different impacts on the regret in different phases:
	\begin{itemize}
	    \item \textbf{Exploration phase.} The loss-$1$ attack does not negatively affect the exploration phase if the preceding communication phase is successful, as it would not ruin the synchronization among players. Especially, with successful synchronization, strategies can be designed to better allocate arms and minimize the regret during exploration phases.
	    \item \textbf{Communication phase.} The loss-$1$ attack in a communication phase requires special attention, as it may lead to communication errors for players, which jeopardize the essential synchronization among them and lead to a potential \emph{linear} regret due to collisions in the subsequent exploration phase, as illustrated in Fig.~\ref{fig:phase}.
	\end{itemize}
	Because of these different impacts, the worst-case adversarial attack scenario can have all-one loss sequences for all communication phases, which prevents any information sharing among players and may cause a linear loss over the entire $T$ time slots.

	\begin{figure}[bth]
		\centering
		\includegraphics[width=0.9 \linewidth]{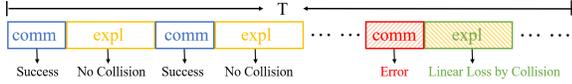}
		\caption{Illustration of communication and exploration phases with communication errors.}
		\label{fig:phase}
	\end{figure}

	From the previous studies in the stochastic setting \cite{boursier2019sic,Shi2020aistats}, it is clear that any bandit policy attempting to enable forced-collision communication in the no-sensing setting will have a dependency on the environment's ability to ``attack'' such communications. This naturally requires bounding the worst-case loss in communications. In the stochastic settings, this ability is characterized by a positive lower bound $\mu_{\min}$ such that $0<\mu_{\min}\leq \min_{k\in[K]}\mu_{k}$, where $\mu_k$ is the mean of arm $k$'s rewards. For example, such a lower bound is assumed to exist and be known to all players in \cite{boursier2019sic,Shi2020aistats}.

	Analogous to the role of $\mu_{\min}$ in the stochastic MP-MAB models, we propose a new metric to characterize the adversarial environment, called the adversary's \emph{attackability}, which represents the upper bound on the adversary's mechanism in generating loss $1$'s. This is a notable distinction to \cite{bubeck2020non} where no communication is utilized in the no-sensing setting, and thus modeling the adversary's attackability is not necessary. More specifically, in this work, we define two types of attackabilties: the \emph{local attackability} and the \emph{global attackability}, which provide two different ways to \emph{categorize} all the adversaries as detailed in this section.

	First, the local attackability aims at modeling the one-time worst-case attack. In a communication phase shown in Fig.~\ref{fig:phase}, the worst case is for this phase to see all loss $1$'s from the adversary, because no information can be reliably shared in such situation. In other words, the local attackability is captured by the maximum length of contiguous loss $1$'s assigned on the loss sequence, since it represents the longest duration that no reliable communication can happen. Without loss of generality, the local attackability is defined as follows.

	\begin{definition}\label{def:single_attack}
		For time horizon $T$, the \emph{local attackability} $W(T)$ of the adversary is defined as
		 $W(T) = \max_{k\in[K]}n_1^k(T)$,
		where $n_1^k(T)$ denotes the maximum length of the all-one loss sequences that are assigned by the adversary on arm $k$ throughout the $T$ time slots.
	\end{definition}

	It is clear that every possible adversary is featured with a $W(T)\in [0,T]$. In the case that the adversary has $W(T)=\Omega(T)$, it may (asymptotically) attack at all time slots. If $W(T)$ of an adversary is a constant independent of $T$, i.e. $O(1)$, the attack is finite each time. A more general case lies in between these two extremes and any possible adversary can be characterized by the local attackability parameter $\alpha$ as follows.
	\begin{corollary}\label{asp:single_attack}
		For any given adversary, there exists $\alpha \in [0,1]$ such that the local attackability of this adversary satisfies $W(T)\leq O\left(T^{\alpha}\right)$.
	\end{corollary}
	The local attackability captures the one-time ``budget'' for the adversary attack. The adversaries sharing the same parameter $\alpha$ can be viewed as in the same category.

	Another perspective is to consider the overall attacks over $T$ as the \emph{global attackability}. It captures the total amount of loss $1$'s assigned on one arm, and is defined as follows.

	\begin{definition}\label{def:overall_attack}
		For given time horizon $T$, the \emph{global attackability} $V(T)$ of the adversary is defined as
		   $V(T) = \max_{k\in[K]}N_1^k(T)$,
		where $N_1^k(T)$ represents the total number of losses $1$ that are assigned by the adversary on arm $k$ throughout the $T$ time slots.
	\end{definition}
	Similar to $W(T)$, each possible adversary is featured with a $V(T)\in [0,T]$. If an adversary has $V(T)=\Omega(T)$, similar to $W(T)=\Omega(T)$, it may (asymptotically) attack at all time slots and no successful communication can happen. If $V(T)$ of an adversary is a constant independent of $T$, i.e., $O(1)$, the overall attacks are finite, which is negligible asymptotically and equivalent to the collision-sensing setting. Similarly, the adversaries can also be categorized with their global attackability parameter $\beta$ as follows.

	\begin{corollary}\label{asp:overall_attack}
		For any given adversary, there exists $\beta \in [0,1]$ such that the global attackability of this adversary satisfies $V(T) \leq O\left(T^{\beta}\right)$.
	\end{corollary}
	Similar to local attackability, the adversaries sharing the same parameter $\beta$ can be viewed as in the same category, where the overall ``budget'' for the adversary attacks is of the same order. However, note that since the local attackability parameter $\alpha$ in Corollary \ref{asp:single_attack} does not provide any bound on the overall attack budget, it is more stringent than the global attack parameter $\beta$ in Corollary \ref{asp:overall_attack}. For an adversary satisfying a sublinear local attackability, it is possible to have a global attackability of $\Omega(T)$. For example, the adversary with loss sequences of $0,1,0,1, \cdots$ for all arms satisfies the local attackability of $W(T)=O(1)$, i.e., $\alpha = 0$, but also satisfies the global attackability of $V(T)=\Omega(T)$, i.e., $\beta=1$. Finally, we note that the following sections consider either local or global attackability, but not simultaneously.

	It is important to keep in mind that Corollaries~\ref{def:single_attack} and \ref{def:overall_attack} represent two ways of categorizing the adversaries rather than imposing constraints or requirements on them. Each category still has many adversaries as long as their scalings of attackability are the same, and every possible adversary is in a certain category. In addition, as shown in the subsequent sections, such categorization does not even need to be aware by the players. In this scenario, we do not impose more assumptions than \cite{bubeck2020non}. Rather, the attackability view represents a different angle of the {\em same} no-sensing adversarial MP-MAB problem, and the proposed algorithms can adapt to the varying attackability in an automatic way, based on the perceived category of the adversary that it faces.

	A final remark is that the proposed concepts of attackability are closely related to the different modes of external interference sources in the application of cognitive radio. Specifically, the concept of local attackability is suitable for the low-duty-cycle interference source as it captures the essence of burst interference. On the other hand, global attackability focuses on measuring the overall interference and thus is more appropriate to model an interference source with a median or high duty cycle, or the cumulative interference from many low-duty-cycle sources.

	\section{Algorithm Outline}\label{sec:out}
	All the algorithms proposed in this paper have two different phases: exploration phases and communication phases, and share a common leader--follower structure \cite{boursier2019practical,alatur2020multi}. Player $1$ (leader) determines arm assignments for the remaining players (followers). The arm assignment is transmitted to each follower in the communication phases. Then, in the following exploration phases, all the players keep sampling the assigned arms. This section introduces the arm assignment procedure for the exploration phases and gives a brief introduction of the communication phases, which will be separately discussed in the following sections for different attackability scenarios.

	\subsection{Exploration Phase}\label{sec:arm_assign}
	Assuming explicit communications are allowed, i.e., in the centralized model, the challenge of exploration phases is how to choose $M$ arms to explore for all the $M$ players. We note that this is similar to the adversarial multi-play problem, where the leader (the centralized agent) chooses $M$ arms $A(t)=\{A_1(t),...,A_M(t)\}$ at each time step $t$ for the followers to explore, i.e., player $m$ is assigned with arm $A_m(t)$. As commonly adopted in the multi-play setting \cite{uchiya2010algorithms}, each subset of $M$ distinct arms $\{A_1,...,A_M\}$ is viewed as a single meta-arm $A$ to be chosen by the leader. The set of all meta-arms $\mathcal{K}$ is defined as:
	$$
	\mathcal{K}:=\left\{\{A_1,...,A_M\}\subseteq[K]|A_m \not=A_n \text{ for any } m\not=n \right\}.
	$$

	An arm assignment policy that builds on \cite{alatur2020multi} is proposed in this paper. At time $t-1$, players first explore the assigned arms in $A(t-1)$. Then, the leader updates an unbiased loss estimator for each arm $k\in[K]$ as $\tilde{l}_k(t-1)=M  \frac{\hat{l}_{A_1(t-1)}(t-1)}{\sum_{A:k\in A\in \mathcal{K}}P_A(t-1)}  \mathds{1}\{k=A_1(t-1)\}$, where $P_A(t-1)$ is the probability that meta-arm $A$ is chosen at time $t-1$ and $\hat{l}_{A_1(t-1)}(t-1)$ is the loss that the leader observes on her arm $A_1(t-1)$ at time $t-1$. Note that the update only requires past observations from the leader, which is designed to reduce the communication burden and to facilitate the generalization to the decentralized setting\footnote{We note that the recent advance \cite{proutiere2019optimal} in stochastic MP-MAB proves that using information collected by only one leader is sufficient to have an optimal regret behaviour. However, it is unclear whether the similar argument holds for adversarial MP-MAB, which may be an interesting direction for future research.}.

	For time $t$, the cumulative loss estimator $\tilde{L}_A(t)$  for each meta-arm $A\in\mathcal{K}$ is first updated as the sum of the loss estimations of its elementary arms up to time $t-1$, i.e., $\tilde{L}_A(t)=\sum_{\upsilon=1}^{t-1}\sum_{k\in A}\tilde{l}_k(\upsilon)$. Then, the EXP3 algorithm \cite{auer2002nonstochastic} is applied to the meta-arm MAB problem, so that each meta-arm $A\in\mathcal{K}$ is sampled with a probability $P_A(t)$ which is proportional to $\exp(-\eta\tilde{L}_A(t))$, as the exploration meta-arm $A(t)$ for time slot $t$. The loss estimator $\{\tilde{l}_k(t)\}_{k\in[K]}$ is then again updated after pulling the chosen meta-arm. At time $t+1$, the same procedures are performed to get $\{\tilde{L}_A(t+1)\}_{A\in\mathcal{K}}$, $\{P_A(t+1)\}_{A\in\mathcal{K}}$ and $A(t+1)$. As shown in \cite{alatur2020multi} and Appendix \ref{app:cent}, this algorithm guarantees a regret bound of $2M\sqrt{K\log(K)T}$ when $\eta=\sqrt{\log\tbinom{K}{M}/MKT}$.

	Furthermore, since there are $|\mathcal{K}|=\tbinom{K}{M}$ meta-arms, computing the probability $P_A(t)$ for each meta-arm and the marginal probability $\sum_{A:k\in A\in\mathcal{K}}P_A(t)$ for an arm $k$ that is to be updated would lead to an exponential complexity if it is done naively. However, with a concept called K-DPPs \cite{taskar2013determinantal}, sampling and marginalization can be made more efficient. As shown in\cite{alatur2020multi}, the complexity of sampling a meta-arm and computing the marginal probability for a fixed arm can both be reduced to $O(KM)$ with K-DPPs, which makes it less complex for implementation.

	Lastly, with the centralized algorithm described above, the key adjustment to the decentralized setting is to notify followers of their assigned arms by forced-collision communications. However, to avoid a linear communication regret due to frequently updating, the exploration phase is extended from one time slot to $\tau$ slots. This means each player is fixated on one arm for at least $\tau$ slots, and the update happens only after each exploration phase. The leader then uses her samples of losses observed during this entire exploration phase as the feedback to assign arms for the next phase. Note that although this infrequent switching reduces the communication burden, it also degrades the regret guarantees \cite{arora2012online}; we will elaborate on this aspect in the analysis. When there is no ambiguity, the time variable in $A(t)$, $P_A(t)$, $\hat{l}_{k}(t)$, $\tilde{l}_{k}(t)$ and $\tilde{L}_A(t)$ are replaced by the corresponding phase index as $A(p)$, $P_A(p)$, $\hat{l}_{k}(p)$, $\tilde{l}_{k}(p)$ and $\tilde{L}_A(p)$ under the decentralized setting for the $p$-th phase.

	\subsection{Communication Phase}
	In the communication phases, arm assignments are transmitted from the leader to the followers with forced collisions. Functions \texttt{Send()} and \texttt{Receive()} are used in the algorithm description for the sending and receiving procedure with forced collisions. Every player is first assigned with a unique communication arm corresponding to her index, i.e., arm $m$ for player $m$. In the collision-sensing setting, players only need to take predetermined turns to communicate by having the ``receive'' player sample her own communication arm and the ``send'' user either pull (create collision; bit $1$) or not pull (create no collision; bit $0$) the receive player's communication arm to transmit one-bit information. For a player that is not engaged in the current peer-to-peer communication, she keeps pulling her communication arm to avoid interrupting other ongoing communications. Since the collision indicator is perfectly known in the collision-sensing setting, player $m$ can receive error-free information after implicit communication.

	In the more challenging no-sensing setting, there is no information about the collision indicator, which means attacks, i.e., loss $1$ assigned by the adversary, may cause communication errors and incur a linear regret in the subsequent exploration phase as shown in Fig.~\ref{fig:phase}. The no-sensing settings are discussed under four different scenarios (two with knowledge of attackability while the other two without such knowledge) in the following sections, respectively\footnote{The discussions focus on the no-sensing setting, but the proposed algorithms can be easily modified to cover the collision-sensing setting, with details provided in Appendix~\ref{app:singleplayer}.}.
	\begin{itemize}
	    \item \textbf{$\alpha$-aware and $\beta$-aware.} Section \ref{subsec:a_aware} (resp. \ref{subsec:b_aware}) proposes the $\alpha$-aware (resp. $\beta$-aware) A2C2 algorithm with Corollary \ref{asp:single_attack} (resp. \ref{asp:overall_attack}) and players having knowledge of $\alpha$ (resp. $\beta$).
	    \item \textbf{$\alpha$-unaware and $\beta$-unaware.} Similar to the above but players have no knowledge of $\alpha$ (resp. $\beta$). These are reported in Section \ref{subsec:a_unaware} and Section \ref{subsec:b_unaware}, respectively. This is the more challenging case, and the main focus of this work.
	\end{itemize}

	\section{Attackability Known to Players}
	\label{sec:kno}
	In this section, the attackability is assumed to be perfectly known by all players, but such information does not tell the players how and when the attacks would happen. The two definitions of attackability lead to two algorithms, which also serve as the building blocks of subsequent algorithm designs on the attackability-unaware setting.

	\subsection{$\alpha$-aware}\label{subsec:a_aware}
	Although the local attackability is theoretically more stringent than the global attackability, it is relatively easier to handle.
	The $\alpha$-aware A2C2 is presented in Algorithm \ref{alg:a_aware_leader} (leader) and Algorithm \ref{alg:a_aware_follower} (follower). Two information-theoretic concepts that are intimately related to reliable communications but less utilized in the bandit literature are introduced; similar ideas are also adopted in algorithm designs under other scenarios.

	\textbf{Z-channel model.} There exists an asymmetry in collision communication of MP-MAB: bit $1$ (collision) is always received correctly, while bit $0$ (no collision) can be potentially corrupted by a loss $1$ from the adversary. In other words, the adversary attack is \emph{asymmetric} -- she can attack bit 0 but not bit 1. From an information-theoretic point of view, this corresponds to a {Z-channel} model \cite{tallini2002capacity} as shown in Fig.~\ref{fig:z_channel}. We note that this connection to the Z-channel model was first utilized in \cite{Shi2020aistats} to study stochastic no-sensing MP-MAB. A key challenge in the adversary setting as compared to \cite{Shi2020aistats}, however, is that a \emph{fixed} crossover probability does not exist.

    \begin{figure}
	\centering
	\includegraphics[width=0.9\linewidth]{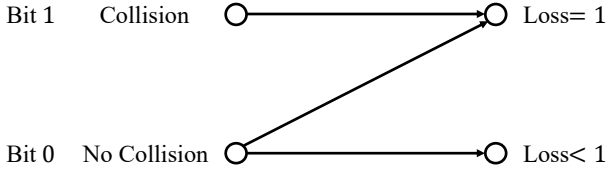}
	\caption{The Z-channel model for implicit communications in no-sensing adversarial MP-MAB.}
	\label{fig:z_channel}
    \end{figure}

	\textbf{Error-correction code with long blocklength.}  The idea of utilizing error-correction code naturally arises under the formulation of a Z-channel model. With the knowledge of $\alpha$, the key idea to overcome the full attackability is to ``overpower'' the adversary via codes that have sufficient error-correction capabilities \cite{LinBook}. Different error-correction codes, ranging from simple repetition code to more complex (and powerful) algebraic and nonlinear codes, can be adopted in the proposed algorithms. To facilitate the regret analysis, we choose to use the repetition code \cite{chen2013optimal} and functions \texttt{rEncoder()} and \texttt{rDecoder()} are used in the algorithms as the encoder and decoder\footnote{Note that more advanced codes can be used, but our regret analysis shows that the regret scaling is not affected.}. At the encoder, each information bit is expanded to a string of length $h(T,\alpha+\epsilon)=\Theta(T^{\alpha+\epsilon})=\omega(T^{\alpha})$, where $\epsilon > 0$ is a fixed constant which can be arbitrarily small\footnote{$f(T)=\Theta(g(T))$ iff $f(T)=O(g(T))$ and $f(T)=\Omega(g(T))$; $f(T)=\omega(g(T))$ iff $\liminf_{T\to\infty}\frac{f(T)}{g(T)}=\infty$.}. Then, the coded bits are sent via forced collisions. If the received bit sequence of length $h(T,\alpha+\epsilon)$ is all-one, then the decoder outputs bit $1$. Otherwise, it outputs bit $0$ (i.e., as long as there exists at least one received bit 0 in the sequence). With $h(T,\alpha+\epsilon) = \omega(T^{\alpha})$, which implies $h(T,\alpha+\epsilon) = \omega(W(T))$, this error-correction code is guaranteed to overpower the local attackability \emph{asymptotically}. Thus, all the communications are guaranteed to be successful \emph{asymptotically}.

		\begin{algorithm}[htb]
		\small
		\caption{$\alpha$-aware A2C2: Leader}
		\label{alg:a_aware_leader}
		\begin{algorithmic}[1]
			\Require $M$, $K$, $T$
			\State \textbf{Initialize:} $\tau\gets \lceil M^{\frac{2}{3}}K^{-\frac{1}{3}}\log(K)^{\frac{1}{3}}T^{\frac{1+2\alpha+2\epsilon}{3}}\rceil; \eta\gets\sqrt{\log\tbinom{K}{M}\tau/ MKT}$
			\For{$p=1,2,...$}
			\State $\forall A\in\mathcal{K}, \tilde{L}_{A}(p)\gets\sum_{v=1}^{p-1}\sum_{k\in A}\tilde{l}_k(v)$
			\State $\forall A\in\mathcal{K}, P_A(p)\gets\frac{e^{-\eta \tilde{L}_A(p)}}{\sum_{J\in\mathcal{K}}e^{-\eta \tilde{L}_{J}(p)}}$\Comment{Loss estimator}
			\State Choose $A(p)=\{A_1(p),...,A_M(p)\}$ with $P_A(p)$
			\State Randomly permute $A(p)$ into $\tilde{A}(p)$
			\Statex $\triangleright$ \textbf{Communication Phase:}
			\State $\forall m\in [M]$, $\text{msg}_m\gets \text{rEncoder}(\tilde{A}_m(p), h(T,\alpha+\epsilon))$
			\State $\forall m\in [M]$, Send$\left(m,\text{msg}_m\right)$\Comment{Send Assignment}
			\Statex $\triangleright$ \textbf{Exploration Phase:}
			\State Stay on arm $\tilde{A}_1(p)$ for $\tau$ time steps
			\State Record cumulative loss $\hat{l}_{\tilde{A}_1(p)}(p)$\Comment{Exploration}
			\State $\forall k\in [K]$, $\tilde{l}_k(p)\gets\frac{M}{\tau}\frac{\hat{l}_{\tilde{A}_1 (p)}(p)}{\sum_{A: k\in A\in \mathcal{K}}P_{A}(p)}\mathds{1}\{\tilde{A}_1 (p)=k\} $
			\EndFor
		\end{algorithmic}
	\end{algorithm}
	\vspace{-0.2in}
	\begin{algorithm}[htb]
		\small
		\caption{$\alpha$-aware A2C2: Follower }
		\label{alg:a_aware_follower}
		\begin{algorithmic}[1]
			\Require{$M$, $K$, $T$, index $m$}
			\State \textbf{Initialize:} $\tau\gets \lceil M^{\frac{2}{3}}K^{-\frac{1}{3}}\log(K)^{\frac{1}{3}}T^{\frac{1+2\alpha+2\epsilon}{3}}\rceil; \eta\gets\sqrt{\log\tbinom{K}{M}\tau/MKT}$
			\For{$p=1,2,...$}
			\Statex $\triangleright$ \textbf{Communication Phase:}
			\State $\text{msg}_m\gets \text{Receive}\left(h(T,\alpha+\epsilon)\right)$\Comment{Receive Assignment}
			\State $\tilde{A}_m(p)\gets$ rDecoder$\left(\text{msg}_m, h(T,\alpha+\epsilon)\right)$
			\Statex $\triangleright$ \textbf{Exploration Phase:}
			\State Stay on arm $\tilde{A}_m(p)$ for $\tau$ time steps \Comment{Exploration}
			\EndFor
		\end{algorithmic}
	\end{algorithm}

	\subsection{$\beta$-aware}\label{subsec:b_aware}
	In the $\alpha$-aware case, although the local attackability is bounded, the adversary can attack \emph{arbitrarily many} times. This is the fundamental reason why each of the communication phases needs protection. However, with an upper bound of global attackability, the adversary has an overall budget for attacking rather than a one-time budget. The $\alpha$-aware A2C2 algorithm can still be applied by replacing $\alpha$ with $\beta$, but it is an overkill to prevent the global attackability in {\em every} communication phase.

	A more efficient algorithm, $\beta$-aware A2C2, is proposed with a holistic consideration for the total budget of the adversary. Details of $\beta$-aware A2C2 can be found in Appendix \ref{app:b_aware}, and we only highlight the key design philosophies here. Because of the iteration between exploration and communication phases, if the adversary succeeds in attacking one communication phase, the immediately following exploration phase of length $\tau$ could be out of synchronization. However, after time $\tau$, players enter the communication phase again and a new iteration starts. In other words, attacks in one communication phase can only cause a one-time linear loss in the next (finite) $\tau$ time slots, while the adversary has a reduced total budget for future attacks. By the time that the adversary runs out of budget, no more communication errors can happen. Thus, as long as the loss caused by the global attackability does not dominate the regret, a certain amount of errors are tolerable. This is a key observation because it means that the power of error-correction coding does not need to be too strong. Technically, instead of coding with a long blocklength of $h(T,\alpha+\epsilon)=\Theta(T^{\alpha+\epsilon})$, a much shorter coded length of $k(T,\nu)=\Theta(T^{\nu})$, where $\nu=\max\{\frac{3\beta-1}{2}, 0\}$, is sufficient of achieving a sublinear regret that is better than $\alpha$-aware A2C2. In fact, a closer look reveals a very surprising result: for $\beta\leq \frac{1}{3}$, we have $\nu=0$, which means there is no need for coding at all in communication phases (see Section~\ref{sec:theory} for details).

	\section{Attackability Unknown to Players}
	\label{sec:unkno}
	In this section, the assumption of the knowledge of attackability is removed. No information on the parameter $\alpha$ or $\beta$ is revealed to the players. The $\alpha$-unaware and $\beta$-unaware settings are tackled separately in the subsequent subsections.

	\subsection{$\alpha$-unaware}\label{subsec:a_unaware}

	With no information of $\alpha$, the main difficulty lies in how to prepare for the worst case without incurring a linear loss. All the key features in $\alpha$-aware A2C2 are still applied in the adaptive algorithm called $\alpha$-unaware A2C2, but several new ideas are needed: an error-detection code to estimate $\alpha$, and a synchronization procedure with randomized length to synchronize the estimation update among players. The algorithms for the leader and followers are presented in Algorithms \ref{alg:a_unaware_leader} and \ref{alg:a_unaware_follower}, respectively.

	\textbf{Estimation of $\alpha$.} Without the knowledge of $\alpha$, no effective prevention is possible for communications in the worst-case scenario. We propose to adaptively estimate $\alpha$ in an escalation fashion. The interval $[0,1]$ (support of $\alpha$) is uniformly divided into sub-intervals with length $\epsilon$, where $\epsilon>0$ is an arbitrarily small constant. The estimated value $\alpha'$ starts with $\alpha'=0$, and increases with a step size of $\epsilon$ while a communication failure is observed until the upper limit is reached. As we see in the regret analysis, this seemingly naive estimation works very well.

	\begin{algorithm}[htb]
	\small
	\caption{$\alpha$-unaware A2C2: Leader}
	\label{alg:a_unaware_leader}
	\begin{algorithmic}[1]
		\Require $M$, $K$, $T$
		\State \textbf{Initialize:} $\alpha$ estimation: $\alpha'\gets 0$; error flag: $F\gets0$
		\For{$p=1,2,...$}
		\State $\tau\gets \lceil M^{\frac{2}{3}}K^{-\frac{1}{3}}\log(K)^{-\frac{1}{3}}T^{\frac{2+\alpha'}{3}}\rceil$
		\State $\eta\gets\sqrt{\log\tbinom{K}{M}\tau/MKT}; \xi\gets\frac{1-\alpha'}{2}$; $F\gets 0$
			\State $\forall A\in\mathcal{K}, \tilde{L}_{A}(p)\gets \sum_{v=1}^{p-1}\sum_{k\in A}\tilde{l}_k(v)$
			\State $\forall A\in\mathcal{K}, P_A(p)\gets \frac{e^{-\eta \tilde{L}_A(p)}}{\sum_{J\in\mathcal{K}}e^{-\eta \tilde{L}_{J}(p)}}$ \Comment{Loss estimator}
			\State Choose $A(p)=\{A_1(p),...,A_M(p)\}$ with $P_A(p)$
			\State Randomly permute $A(p)$ into $\tilde{A}(p)$
			\Statex $\triangleright$ \textbf{Communication Phase:}
			\State $\forall m\in [M]$, $\text{msg}_m\gets \text{eEncoder}(\tilde{A}_m(p), h(T,\alpha'))$
			\State $\forall m\in [M]$, Send$\left(m, \text{msg}_m\right)$ \Comment{Send Assignment}
			\For{$q=1,2,..., N(\xi)$}\Comment{Synchronization}
			\State $\text{msg}_F\gets \text{Receive}\left(h\left(T,\alpha'\right)\right)$
			\State $F \gets$ rDecoder$\left(\text{msg}_F, h\left(T,\alpha'\right)\right)$ \Comment{Uplink: $F=0/1$}
			\State $\text{msg}_F\gets \text{rEncoder}\left(F, h(T,\alpha')\right)$
			\State $\forall m\in [M]$, Send$\left(m,\text{msg}_F\right)$\Comment{Downlink: $F=0/1$}
			\EndFor
			\State $\alpha'\gets\alpha'+F\epsilon$\Comment{Update Estimation}
			\Statex $\triangleright$ \textbf{Exploration Phase:}
			\State Stay on arm $\tilde{A}_1(p)$ for $\tau$ time steps\;
			\If{$F=0$}
			\State Record cumulative loss $\hat{l}_{\tilde{A}_1(p)}(p)$
			\State
				$\forall k\in[K], \tilde{l}_k(p)\gets\frac{M}{\tau}\frac{\hat{l}_{\tilde{A}_1 (p)}(p)}{\sum_{A:k\in A\in \mathcal{K}}P_A(p)} \mathds{1}\{\tilde{A}_1 (p)=k\}$
			\Else \ $\forall k\in[K], \tilde{l}_k(p)\gets 0$
			\EndIf
		\EndFor
	\end{algorithmic}
	\end{algorithm}

	\begin{algorithm}[htb]
	\small
	\caption{$\alpha$-unaware A2C2: Follower}
	\label{alg:a_unaware_follower}
	\begin{algorithmic}[1]
		\Require $M$, $K$, $T$, index $m$
		\State \textbf{Initialize:} $\alpha$ estimation: $\alpha'\gets 0$; error flag: $F\gets0$; exploration set: $\mathcal{S}\gets \emptyset$
		\For{$p=1,2,...$}
		\State $\tau\gets \lceil M^{\frac{2}{3}}K^{-\frac{1}{3}}\log(K)^{-\frac{1}{3}}T^{\frac{2+\alpha'}{3}}\rceil$
		\State $\eta\gets\sqrt{\log\tbinom{K}{M}\tau/MKT}$; $\xi\gets\frac{1-\alpha'}{2}$; $F\gets 0$
			\Statex $\triangleright$ \textbf{Communication Phase:}
			\State $\text{msg}_m \gets \text{Receive}\left(h(T,\alpha')\right)$
			\State $\mathcal{S}\gets$ eDecoder$\left( \text{msg}_m, h(T,\alpha')\right)$\Comment{Receive Assignment}
			\State Randomly choose in $\mathcal{S}$ for $\tilde{A}_m(p)$
			\State $F \gets \mathds{1}\left\{|\mathcal{S}|>1\right\}$\Comment{Comm Error or not}
			\For{$q=1,2,..., N(\xi)$}
				\State $\text{msg}_F\gets \text{rEncoder}\left(\text{F}, h(T,\alpha')\right)$
				\State Send$(1, \text{msg}_F)$\Comment{Uplink: $F=0/1$}
				\State $\text{msg}_F\gets \text{Receive}\left(h(T,\alpha')\right)$
				\State $F \gets \text{rDecoder}\left(\text{msg}_F, h(T,\alpha')\right)$\Comment{Downlink: $F=0/1$}
			\EndFor
			\State $\alpha'\gets\alpha'+F\epsilon$\Comment{Update Estimation}
			\Statex $\triangleright$ \textbf{Exploration Phase:}
			\State Stay on arm $\tilde{A}_m(p)$ for $\tau$ time steps
		\EndFor
	\end{algorithmic}
	\end{algorithm}

	\textbf{Error-detection repetition code.} The aforementioned escalation mechanism to estimate $\alpha$ relies on knowing when communication failure happens, which is non-trivial. This leads to the second idea of utilizing a special kind of error-detection code for the Z-channel, called the constant weight code \cite{borden1982optimal}. Codewords in one constant weight code share the same Hamming weight, which enables error detection. As noted in \cite{borden1982optimal}, a constant weight code can detect any number of asymmetric errors, and the maximal number of constant-weight codewords of length $n$ can be attained by taking all codewords of weight $\lceil\frac{n}{2}\rceil$ or $\lfloor\frac{n}{2}\rfloor$. Thus, codeword length of $O(\log(K))$ is theoretically sufficient to enable a constant-weight code with $O(K)$ codewords\footnote{This observation can be verified by invoking the Stirling's formula: $n!\approx \sqrt{2\pi n}\left(\frac{n}{e}\right)^{n}$.}.

	To facilitate the discussion, a particular kind of constant weight code is adopted. As shown in Fig.~\ref{fig:ed_code}, while transmitting arm $k$'s index from the leader, it is represented by a bit sequence of length $K$ in which the $k$-th bit is $1$ while all other bits are $0$. Then, each bit of this sequence is repeated $h(T,\alpha')=\Theta(T^{\alpha'})$ times. Thus, all the codewords share the weight of $h(T,\alpha')$. The resulting coded bit sequence is then transmitted with forced collisions. Upon receiving, the entire bit string is divided into $K$ blocks and each block is processed separately. The decoder outcome $\mathcal{S}$ is the (possibly multiple) indices of the blocks that have all-ones.

	This is a very effective error \emph{detection} method because, based on the property of the Z-channel, the source index $k$ is always decoded correctly. Thus, if the decoder outputs more than one indexes, there must be a communication error (example `F' in Fig.~\ref{fig:ed_code}), meaning the current communication protocol with $\alpha'$ is not strong enough to overcome the attacks. Otherwise, there is only one index as the decoder output, which means the communication is successful (example `S' in Fig.~\ref{fig:ed_code}). This suggests that it is sufficient to maintain the current estimation $\alpha'$. Once $\alpha'= v\epsilon> \alpha$, we have $h(T,\alpha')=\Theta(T^{\alpha'})=\omega(W(T))$, which means there will be no more communication errors (asymptotically) and $\alpha'$ can be used for the remainder of time slots. Functions $\texttt{eEncoder()}$ an $\texttt{eDecoder()}$ are used in the algorithms for the corresponding error-detection encoder and decoder, respectively.

    \begin{figure}
	\centering
	\includegraphics[width=\linewidth]{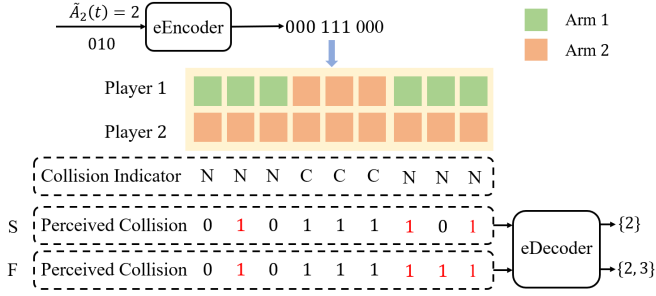}
	\caption{Example of one communication phase and error-detection coding with $M=2$, $K=3$ and $h(T, \alpha')=3$.}
	\label{fig:ed_code}
    \end{figure}

	\textbf{Synchronization with randomized length.} Error-detection repetition code allows each follower to decide whether the $\alpha$-estimation needs to be updated. However, the leader does not have access to this information and such estimation \emph{across} players may not be synchronized, which poses a significant challenge that calls for communication for synchronization. Note that unlike communications for arm assignment, not all errors in this procedure are equally bad. The worst case during synchronization is \emph{uneven attacks}: attacks that happen only on a subgroup of players, i.e., some players receive incorrect signals and update the $\alpha$ estimation, while the others do not. On the other hand, if \emph{all} players receive wrong signals simultaneously, they are still synchronous because they all will increase $\alpha'$ to the next value.

	To solve this problem, we return to the fundamental technique of (single-player) adversarial bandits and introduce \emph{randomness} to the synchronous procedure. After communication of arm assignment in each round, the followers report to the leader whether communication errors occurred in this round at the same time. The followers send bit $1$ representing error and bit $0$ representing no error with the same error-correction repetition code of length $h(T,\alpha')$. Borrowing the terminology from wireless communications, this process is referred to as `uplink', and as long as one follower has communication errors, the leader can correctly receive bit $1$ (signal for updating). The uplink is robust to the attacks since there is only one receiver (the leader); even if an attack is successful, the need for updating is still conveyed correctly. After the uplink, if bit $1$ (no matter from followers or the adversary) is received by the leader, she sends back bit $1$ to every follower; otherwise, bit $0$ is sent. This is called `downlink'\footnote{Note that `uplink' denotes that followers communicate to the leader, while `downlink' denotes the opposite direction, i.e., the leader communicates to the followers.}. If bit $1$ is sent, all the followers receive it correctly over the Z-channel. However, in the case of bit $0$, uneven attacks on partial followers may happen, which leads to a failed synchronization. To keep synchronization successful with a high probability, the `uplink-downlink' procedure keeps iterating for a \emph{random} number of rounds $N(\xi)$, which is generated by a shared randomness source among players\footnote{Specifically, we can assume that players have access to a shared random seed in advance, e.g., the starting time of the game, and then use this random seed to generated $N(\xi)$ for later synchronizations.} and uniformly distributed in $\left[0, \left \lceil T^{\xi} \right \rceil\right]$. If the follower detects communication errors during the assignment or receives bit $1$ in the preceding downlink, she keeps sending bit $1$ to the leader in the following uplink cycles until the procedure ends. For this protocol, the adversary has to \emph{exactly} attack the last round of downlink to destroy synchronization, since attacks at other rounds are broadcasting. As analyzed in Section \ref{sec:theory}, this procedure with a carefully chosen $\xi=\frac{1-\alpha'}{2}$ (which changes with $\alpha'$) is crucial in maintaining a sub-linear regret.

	In addition to updating the estimation, the synchronization also plays an important role in maintaining unbiased loss estimations. With communication errors, potential collisions may happen on the assigned arm of the leader, thus an unbiased estimation is impossible if the leader uses the observed rewards of this round as feedback. The algorithm allows the leader to only use collected cumulative rewards when the signals from followers indicate correct communications.

	With full details given in Algorithms~\ref{alg:a_unaware_leader} and \ref{alg:a_unaware_follower}, we offer a summary overview of the algorithmic structure of $\alpha$-unaware A2C2 as follows.
	\begin{itemize}[leftmargin=*]\itemsep=0pt
	    \item \textbf{Step I: meta-arm selection.} With the arm assignment policy specified in Section~\ref{sec:arm_assign}, a meta-arm is chosen by the leader;
	    \item \textbf{Step II: communication phase for arm assignment.} The leader encodes the indices of the chosen arms with the error-detection code. Then, she informs each follower of her assigned arm by implicitly communicating the coded arm index. Each follower decodes the received message, and detects whether a communication error has occurred.
	    \item \textbf{Step III: communication phase for synchronization.} A random length of ``uplink-downlink'' iteration is performed, which composes of the synchronization procedure. After synchronization, each player updates the estimation of the attackability parameter $\alpha'$.
	    \item \textbf{Step IV: exploration phase.} The players pull the assigned arms for a certain duration. The leader counts the cumulative losses which will be used by the loss estimators.
	\end{itemize}
	Similar structures also apply to the previously discussed $\alpha/\beta$-aware A2C2 algorithms by removing the communication phase for synchronization. For the $\beta$-unaware A2C2 algorithm, detailed changes are discussed in the next subsection.

	\subsection{$\beta$-unaware}\label{subsec:b_unaware}
	Similar ideas of the previous design can be applied to the $\beta$-unaware setting, and the resulting $\beta$-unaware A2C2 algorithm is presented in Appendix~\ref{app:b_unaware}. Some important differences to $\alpha$-unaware A2C2 are explained in this subsection.

	Unlike the $\alpha$-estimation starting from $\alpha'=0$, the estimation $\beta'$ starts with $\frac{1}{4}$ as analyzed in Section \ref{sec:theory}. In each communication phase, the arm assignment is similarly encoded with the error-detection code but the codeword length for each bit is adjusted to $k(T,\nu')=\Theta(T^{\nu'})$, where $\nu'=\frac{4\beta'-1}{3}$. The increased coding rate is due to the same reason described in Section \ref{subsec:b_aware}, i.e., a certain amount of communication errors are tolerable within the global attackability bound.

	As for the feedback from the followers, two different uplink operations are performed: one to maintain the unbiased loss estimations and the other to keep players synchronized with the estimated $\beta'$. First, after each arm assignment, followers immediately notify the leader of communication errors for her to maintain an unbiased loss estimation. This uplink does not indicate the need of updating $\beta'$ since $\beta'$ is an estimation of the \textit{overall} budget and there is no following downlink. Since communication errors in this uplink can only influence the subsequent exploration phase, it is performed with a repetition code of length $k(T,\nu')=\Theta(\nu')$.  Then, for the update of $\beta'$, each player keeps counting the overall number of attacks on her communication arm and reports to the leader when the estimated budget is exceeded. To reduce the communication burden, the update of $\beta'$ and the synchronization procedure is performed only at potential updating time slots rather than after each communication phase. Specifically, similar iterations of uplink and downlink for synchronization happen every $\lceil T^{\beta'}/k(T,\nu')\rceil$ phases, with length $k(T,\beta')=\Theta(T^{\beta'})$ in each round and a random number of rounds $N(\xi)\in \left [0, \left \lceil T^{\xi} \right\rceil \right ]$. The choice of length $k(T,\beta')$ is because that synchronization error may influence the entire remaining time slots. Similar to analysis in Section \ref{subsec:a_unaware}, the adversary must attack \emph{exactly} the last round of synchronization to succeed in breaking the coordination among players.

	\section{Performance Analysis}
	\label{sec:theory}
	This section is devoted to the theoretical analyses for all proposed A2C2  algorithms. Detailed proofs can be found in Appendices \ref{app:a_aware} to \ref{app:b_unaware}.

	\textbf{$\alpha/\beta$-aware.} The regret of $\alpha$-aware A2C2 and $\beta$-aware A2C2 algorithms are first presented in Theorems \ref{thm:a_aware} and \ref{thm:b_aware}, respectively.

	\begin{theorem}[$\alpha$-aware]\label{thm:a_aware}
		With $\tau = \lceil M^{\frac{2}{3}}K^{-\frac{1}{3}}\log(K)^{\frac{1}{3}}T^{\frac{1+2\alpha+2\epsilon}{3}}\rceil$, $\eta=\sqrt{\log\tbinom{K}{M}\tau/MKT}$, the expected regret of $\alpha$-aware A2C2 algorithm is bounded by
		\begin{align*}
		    \mathbb{E}\left [R_1(T) \right ] \leq O \left(M^{\frac{4}{3}}K^{\frac{1}{3}}\log(K)^{\frac{2}{3}}T^{\frac{2+\alpha+\epsilon}{3}}  \right),
		\end{align*}
		where $\epsilon>0$ is an arbitrarily small constant.
	\end{theorem}

	\begin{theorem}[$\beta$-aware]\label{thm:b_aware}
		With $\tau = \lceil K^{\frac{1}{3}}\log(K)^{-\frac{1}{3}}T^{\max\{\beta,\frac{1}{3}\}}\rceil$, $\eta=\sqrt{\log\tbinom{K}{M}\tau/MKT}$, $\nu=\max\{\frac{3\beta-1}{2},0\}$, the expected regret of $\beta$-aware A2C2 algorithm is bounded by
		\begin{align*}
		    \mathbb{E}\left[R_2(T)\right] \leq  O\left(M^{2}K^{\frac{2}{3}}\log(K)^{\frac{1}{3}}T^{\max\left\{\frac{1+\beta}{2}, \frac{2}{3}\right\}}\right).
		\end{align*}
	\end{theorem}

	It is worth noting that for $\beta\leq \frac{1}{3}$, we have $\nu=0$ which means there is no need for coding at all.  Another note is that for $\alpha=0$ or $\beta\leq \frac{1}{3}$, the known regret in collision-sensing setting $O(T^{\frac{2}{3}})$ \cite{alatur2020multi} is recovered. When $\alpha=1$ or $\beta=1$, which means $W(T)=\Omega(T)$ or $V(T)=\Omega(T)$, the adversary can asymptotically attack all time slots to prevent any communication, and thus our regret becomes $O(T)$.

	\textbf{$\alpha/\beta$-unaware.} Without knowledge of the attackability, the performance of $\alpha$-unaware A2C2 and $\beta$-unaware A2C2 algorithms are guaranteed in Theorems \ref{thm:a_unaware} and \ref{thm:b_unaware}, and analyzed subsequently.

	\begin{theorem}[$\alpha$-unaware]\label{thm:a_unaware}
		With $\tau= \lceil M^{\frac{2}{3}}K^{-\frac{1}{3}}\log(K)^{-\frac{1}{3}}T^{\frac{2+\alpha'}{3}}\rceil$,	$\eta=\sqrt{\log\tbinom{K}{M}\tau/MKT}$ and $\xi=\frac{1-\alpha'}{2}$ under the estimation $\alpha'$, the expected regret of the $\alpha$-unaware A2C2 algorithm is bounded by
		\begin{align*}
		    \mathbb{E}[R_3(T)] \leq  O\left(M^{\frac{4}{3}}K^{\frac{1}{3}}\log(K)^{\frac{1}{3}}T^{\frac{5+\alpha+\epsilon}{6}}\right),
		\end{align*}
		where $\epsilon>0$ is an arbitrarily small constant.
	\end{theorem}

	\begin{theorem}[$\beta$-unaware]\label{thm:b_unaware}
		With $\tau= \lceil K^{-\frac{1}{3}}\log(K)^{-\frac{1}{3}}T^{\frac{1+2\beta'}{3}}\rceil$, $\eta=\sqrt{\log\tbinom{K}{M}\tau/MKT}$, $\xi=\frac{1+2\beta'}{3}$ and $\nu'=\frac{4\beta'-1}{3}$ under estimation $\beta'$ starting from $\frac{1}{4}$, the expected regret of the $\beta$-unaware A2C2 algorithm is bounded by
		\begin{align*}
		    \mathbb{E}[R_4(T)] \leq  O\left(M^2K^{\frac{1}{3}}\log(K)^{\frac{1}{3}}T^{\max\left\{\frac{2+\beta+\epsilon}{3},\frac{3}{4}\right\}}\right),
		\end{align*}
		where $\epsilon>0$ is an arbitrarily small constant.
	\end{theorem}

	Similar to the $\beta$-aware case, we have $\nu'=0$ for $\beta\leq \frac{1}{4}$ in the $\beta$-unaware A2C2 algorithm, which indicates there is no need of coding for assigning arms and reporting communication errors. Compared with Theorem \ref{thm:a_aware}, the lack of knowledge of $\alpha$ worsens the regret by a factor of $O(T^{\frac{1-\alpha}{6}})$ in Theorem \ref{thm:a_unaware}. Similarly, the lack of knowledge of $\beta$ degrades the regret by $O(T^{\frac{1-\beta}{6}})$.  It is also worth noting that Theorems \ref{thm:b_aware} and \ref{thm:b_unaware} provide better dependencies on $T$ than Theorems \ref{thm:a_aware} and \ref{thm:a_unaware}, respectively, which reinforces the intuition that the local attackability parameter $\alpha$ in Corollary \ref{asp:single_attack} is more stringent than the global attackability parameter $\beta$ in Corollary \ref{asp:overall_attack}.

	Compared with the regret of $O(MK^{\frac{3}{2}}T^{1-\frac{1}{2M}})$ in \cite{bubeck2020non}, it can be observed that the regret results of A2C2 have an exponential dependence on the attackability rather than the number of players $M$, which could be an advantage while dealing with a large number of players. From another perspective, these two different dependencies reveal two orthogonal ``dimensions of hardness'' in the no-sensing adversarial MP-MAB problem: multiple players and attackability. As no information sharing among players is utilized in \cite{bubeck2020non}, the coordination is limited and the difficulty of the problem grows exponentially with the number of players. In our work, forced collisions are used for communications and coordination among players is established. As a result, the regret shifts the exponential dependence from number of players ($M$) to attackability ($\alpha$ or $\beta$), and the dependence on $M$ is only a multiplicative factor.

	\subsection{Proof sketch: $\alpha$-unaware}
	Regret of the $\alpha$-unaware A2C2 algorithm can be decomposed as $R_3(T)=R_3^{expl}(T)+R_3^{comm}(T)+R_3^{sync}(T)$, corresponding to the exploration regret, communication regret, and synchronization error regret, respectively. The remaining of this subsection is devoted to a brief discussion of each regret term in the decomposition and its corresponding upper bound.

	First, the exploration regret $R^{expl}_3(T)$ can be further divided into two parts depending on whether the preceding communication phase is successful or not. The first part has a regret caused by exploration without collision. The second part is caused by potential collisions due to the failed communications, which only occur with an underestimated attackability. Lemma \ref{lem:a_unaware_expl} presents an upper bound for the overall exploration regret.

	\begin{lemma}\label{lem:a_unaware_expl}
		Denoting $\zeta=\mathds{1}\{$successful preceding communication$\}$, the expected exploration regret of the $\alpha$-unaware A2C2 algorithm is bounded by:
    	\begin{align*}
    	    \mathbb{E}\left[R_3^{expl}(T)\right]&=\mathbb{E}\left[R_3^{expl}(T|\zeta=1)\right]+\mathbb{E}\left[R_3^{expl}(T|\zeta=0)\right] \\
    	    &\leq O\left(M^{\frac{4}{3}}K^{\frac{1}{3}}\log(K)^{\frac{1}{3}}T^{\frac{5+\alpha+\epsilon}{6}}\right).
    	\end{align*}
	\end{lemma}

 	The communication regret of $\alpha$-unaware A2C2 consists of two parts: the first from arm assignments and the second from synchronizations, which can be bounded as follows.
	\begin{lemma}\label{lem:a_unaware_comm}
		The expected communication regret of $\alpha$-unaware A2C2 is bounded by
		\begin{align*}
		    \mathbb{E}\left[R_3^{comm}(T)\right] \leq O\left(M^{\frac{4}{3}}K^{\frac{1}{3}}\log(K)^{\frac{1}{3}}T^{\frac{5+\alpha+\epsilon}{6}}\right).
		\end{align*}
	\end{lemma}

	Finally, the regret $R_3^{sync}(T)$ caused by the risk of losing synchronization during updates is bounded in the following lemma, where the worst case is assumed such that linear regret is caused by collisions. Note that this term is unique in the attackability-unaware setting as there is no synchronization performed in the attackability-aware setting.
	\begin{lemma}\label{lem:a_unaware_sync}
		The expected regret caused by potential synchronization errors of $\alpha$-unaware A2C2 is bounded by:
		\begin{align*}
		\mathbb{E}\left[R_3^{sync}(T)\right] \leq O\left(M^{\frac{1}{3}}K^{\frac{1}{3}}\log(K)^{\frac{1}{3}}T^{\frac{5+\alpha+\epsilon}{6}}\right).
		\end{align*}
	\end{lemma}

	Combining Lemmas \ref{lem:a_unaware_expl} to \ref{lem:a_unaware_sync}, the overall regret in Theorem \ref{thm:a_unaware} is proved. Note that all three component regret bounds share the same order of $T$ and similar factors of $M$ and $K$, which is not a coincidence. As the analysis in Appendix \ref{app:a_unaware} shows, these three terms can be further decomposed based on the estimation $\alpha'$, and they are all dominated by the elements associated with the final (largest) estimation. Optimizations over $\tau$ and $\xi$ are carried out based on these elements, which carefully control the regret associated with the communication error and the synchronization error to avoid one component regret dominating others.

	\subsection{Proof sketch: $\beta$-unaware}
	The regret of $\beta$-unaware A2C2 can be decomposed as $R_4(T)=R_4^{expl}(T)+R_4^{err}(T)+R_4^{comm}(T)+R_4^{sync}(T)$, which refers to the exploration regret with successful preceding communication phases, exploration regret with failed preceding communication phases, communication regret, and synchronization error regret, respectively.
	The decomposed regret components $R_4^{expl}(T)$, $R_4^{comm}(T)$ and $R_4^{sync}(T)$ are similar to the corresponding components in the proof of $\alpha$-unaware A2C2, which are bounded in Lemmas \ref{lem:b_unaware_expl} to \ref{lem:b_unaware_sync}.

	\begin{lemma}\label{lem:b_unaware_expl}
		The expected regret caused by explorations after successful communications in $\beta$-unaware A2C2 is bounded by:
		\begin{align*}
		\mathbb{E}\left[R_4^{expl}(T)\right]\leq O\left(MK^{\frac{1}{3}}\log(K)^{\frac{1}{3}}T^{\max\left\{\frac{2+\beta+\epsilon}{3},\frac{3}{4}\right\}}\right).
	    \end{align*}
	\end{lemma}

	\begin{lemma}\label{lem:b_unaware_comm}
		The expected communication regret of $\beta$-unaware A2C2  is bounded by:
		\begin{align*}
		\mathbb{E}\left[R_4^{comm}(T)\right] \leq O\left(M^2K^{\frac{1}{3}}\log(K)^{\frac{1}{3}}T^{\max\left\{\frac{2+\beta+\epsilon}{3},\frac{3}{4}\right\}}\right).
		\end{align*}
	\end{lemma}

	\begin{lemma}\label{lem:b_unaware_sync}
		The expected regret caused by potential synchronization errors of $\beta$-unaware A2C2  is bounded by:
		\begin{align*}
		\mathbb{E}\left[R_4^{sync}(T)\right] \leq O\left(MK^{\frac{1}{3}}\log(K)^{\frac{1}{3}}T^{\max\left\{\frac{2+\beta+\epsilon}{3},\frac{3}{4}\right\}}\right).
		\end{align*}
	\end{lemma}

    Since $\beta$-unaware A2C2 is designed to be capable of handling a certain amount of communication errors, the exploration regret associated with failed communications (i.e., $R_4^{err}(T)$) is separately analyzed, as opposed to being lump summed into one exploration loss as in the local attackability analysis.
	\begin{lemma}\label{lem:b_unaware_err}
		The expected regret caused by explorations after failed communications in $\beta$-unaware A2C2  is bounded by:
    	\begin{align*}
		\mathbb{E}\left[R_4^{err}(T)\right] \leq O\left( M^2K^{\frac{1}{3}}\log(K)^{-\frac{1}{3}}T^{\max\left\{\frac{2+\beta+\epsilon}{3},\frac{3}{4}\right\}}\right).
		\end{align*}
	\end{lemma}
	In addition to similar optimizations over $\tau$ and $\xi$ as in the analysis of $\alpha$-unaware-A2C2, the choice of $\nu'$ is also optimized so that the exploration regret caused by the global attackability, i.e. $\mathbb{E}[R_4^{err}(T)]$, does not dominate the total regret.

	\section{Discussions}
	\label{sec:dis}
	We discuss some algorithmic details and a few directions for potential future research.

	\mypara{Sublinear regrets.} As shown in Theorems \ref{thm:a_aware} to \ref{thm:b_unaware}, the regrets of the A2C2 algorithms can be sublinear as long as the adversary cannot (asymptotically) attack all time slots, i.e., $\alpha< 1$ or $\beta< 1$, which is similar to the no-sensing stochastic setting that sublinear regrets can be achieved only when $\mu_{\min}>0$. Further, as stated in Section \ref{sec:advattack}, it is possible that $\beta=1$ while $\alpha<1$, and in such cases, $\alpha$-(un)aware A2C2 can still achieve a sublinear regret.

	\mypara{The choice of $\epsilon$.} With an unknown attackability, i.e. $\alpha$ or $\beta$-unaware A2C2, $\epsilon$ can be an arbitrarily small constant in the asymptotic case. In the finite case, this choice influences both the convergence of estimation and the regret behavior. A larger $\epsilon$ can let the algorithm escalate faster but with a potential overestimation and a larger regret. With a smaller $\epsilon$, although the escalation is slower, better performance can be achieved with a more precise estimation.

	\mypara{Implicit communication of loss information.}	As stated in Section~\ref{sec:arm_assign}, the choices of meta-arms in A2C2 algorithms are only based on observations of the leader. While this simplifies communications, an open problem is how to transmit the collected loss information between players through collisions. The methods of arm index sharing proposed in this work is a starting point, but more careful designs may further reduce the regret, e.g., quantizing the losses.

	\mypara{Combining A2C2 and \cite{bubeck2020non}.}	Without a mature method of statistics sharing, another idea to increase the information usage is to have multiple leaders instead of one, which alludes to a hybrid algorithm of this work and \cite{bubeck2020non}. Specifically, each leader with a certain number of followers can form as a meta-player. Among the meta-players, the method from \cite{bubeck2020non} can be adopted since it does not require information sharing. Within each meta-player, A2C2 algorithms can be performed to let the leader coordinate the followers. It would be interesting to see a rigorous analysis of this hybrid algorithm, especially its dependence on  $M$.

	\mypara{Dynamic setting.} While this work assumes the players are fixed in the entire game, it is an interesting (but also very challenging) open problem to consider the dynamic setting, where players can enter and leave the game freely. Some attempts have been made in the stochastic MP-MAB \cite{rosenski2016multi,boursier2019sic} and the collision-sensing adversarial MP-MAB \cite{bande2019adversarial}, while the no-sensing adversarial dynamic environment remains largely open. For the attackability-aware setting (i.e., $\alpha/\beta$-aware A2C2), algorithmic extensions are relatively easy since the new players can be synchronized with the already existing players using the knowledge of attackability, but the theoretical analysis for this case may be difficult. If the new players do not have information of the current estimation of attackability (i.e., $\alpha/\beta$-unaware A2C2), this extension becomes difficult. With some constraint on the frequency of players entering or leaving the game, it is conceivable to run $\alpha/\beta$-unaware A2C2 independently on each of these intervals, which is similar to \cite{rosenski2016multi}. Nevertheless, its regret analysis would be highly nontrivial.

	\section{Experiments}\label{sec:exp}
    More numerical illustrations comparing A2C2 and the algorithm in \cite{bubeck2020non} are provided in Appendix~\ref{apdx:sim}. The A2C2 algorithms are also empirically evaluated against a naive adversarial algorithm as well as the EC-SIC design in \cite{Shi2020aistats}, which is a stochastic no-sensing algorithm. In the naive adversarial algorithm, each player individually runs EXP3 \cite{auer2002nonstochastic} using her own observed loss information, which results in $M$ parallel EXP3 algorithms. The EC-SIC algorithm is the state-of-the-art stochastic no-sensing algorithm, which also utilizes error-correction coding but is confined to stochastic loss distributions. The testing environments have $M=4$ and $K=10$ (same as the numerical illustrations in Fig.~\ref{fig:attack}) and the step-size parameter $\epsilon$ is set as $0.01$. The following figures are plotted by running the algorithms under different time horizons and the regret result of each horizon is the average over $50$ runs.

	The first loss sequence for evaluation is generated similar to \cite{bande2019adversarial}. Specifically, for arm $k$, a random variable $c_k$ is first sampled from the uniform distribution on $[0.2,0.9]$, then $l_k(t)$ is drawn from the uniform distribution on $[c_k,0.9]$. This results in a non-stationary environment. After the generation, to explicitly characterize the adversarial attackability, contiguous loss-$1$ sequences of length $50$ are randomly spread in the whole loss sequence to replace the originally generated losses. Fig.~\ref{fig:nonstationary1} illustrates the performance of $\alpha$-unaware and $\beta$-aware A2C2s against the other two benchmarks. It can be observed that $\alpha$-unaware and $\beta$-aware A2C2s outperform the baselines, and $\beta$-aware A2C2 has the best performance, which corroborates the regret analysis.

	\begin{figure}[htb]
		\centering
		\includegraphics[width=0.9\linewidth]{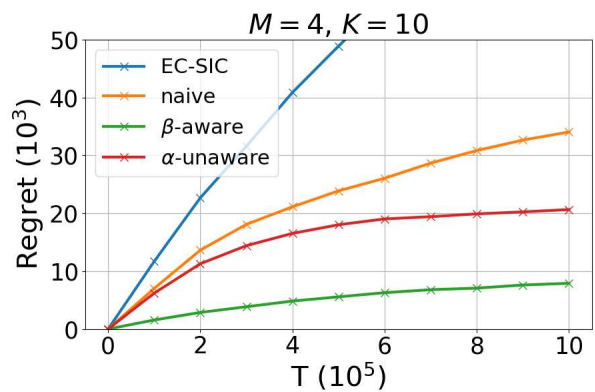}
		\caption{Regret comparison where the loss generations are non-stationary in the entire horizon.}
		\label{fig:nonstationary1}
	\end{figure}
	\begin{figure}[htb]
		\centering
		\includegraphics[width=0.9\linewidth]{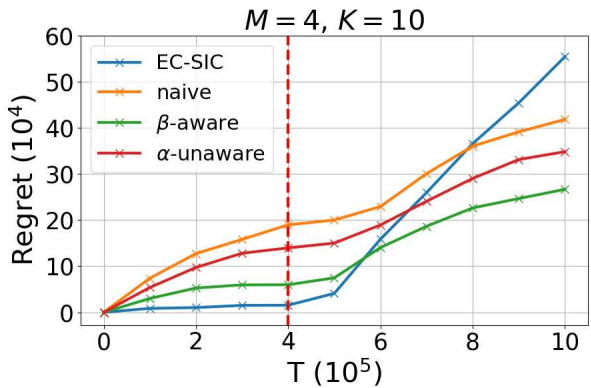}
		\caption{Regret comparison where the loss generations change at $T'=4\times 10^5$ (labeled with the dotted red line).}
		\label{fig:nonstationary2}
	\end{figure}

	The second loss sequence is generated by shifting the mean values of the stochastic loss distribution at a changing point, which is similar to \cite{alatur2020multi}. Specifically, before $T'=4\times 10^5$, the loss expectations are $[a_k]_{k=1}^K = [0.2,0.2,0.2,0.2,0.4,0.4,0.4,0.4,0.4,0.4]$, where arms $1-4$ are better than the others, and the loss for each arm $k$ is sampled from the uniform distribution on $[a_k-0.15,a_k+0.15]$. Then, after $T'$, some shifts happen to the loss expectations as $a_1=a_4=0.8$ and $a_5=a_6=0.2$, thus arms $2,3,5,6$ now generate lower losses. Then, contiguous loss-$1$ sequences of length $50$ are similarly spread into the whole loss sequences. The results are reported in Fig.~\ref{fig:nonstationary2}. It can be observed that while EC-SIC has the best performance at the beginning, its performance quickly degrades after the changing point $T'$ since it has already converged to the originally good arms, some of which (arm $1$ and $4$) however become sub-optimal afterwards. Although the performance of the two A2C2 algorithms also fluctuate after the changing point, they quickly adjust to the new environment and re-gain a sublinear regret behavior.

	\section{Conclusions}
	\label{sec:conc}
	This work made progress in the no-sensing adversarial MP-MAB problem by incorporating implicit communication. We have introduced the concept of \emph{attackability} to categorize all possible adversaries from either a local view or a global view, and designed \emph{Adversary-Adaptive Collision-Communication (A2C2)}, a family of algorithms that can handle known or unknown attackabilities. The algorithmic contributions mainly came from several new tools in information theory and communication theory, such as Z-channel model, error-correction coding, (novel) error-detection repetition code, and uplink-downlink communication with randomized length. Theoretical analyses showed that the proposed algorithms have attackability-dependent regrets, which eliminate the exponential dependence on the number of players in the state-of-the-art no-sensing adversarial MP-MAB research, and revealed a new dimension of hardness, i.e., attackability, that compliments the hardness associated with the number of players.

	\appendix
	\subsection{Special cases: $M=1$ and collision-sensing}\label{app:singleplayer}
	The regrets from communications, synchronization errors, and explorations after failed communications do not exist for $M=1$ (single player); our result recovers $O(T^{\frac{1}{2}})$ in this special case. The reason that the asymptotic regret in Theorems \ref{thm:a_aware} to \ref{thm:b_unaware} cannot recover this special case is that we focus on the scaling behavior, and thus the multiplicative factor $M-1$ is simplified to $M$ in the ensuing analysis.

	We further note that if the collisions can be perceived by players, i.e., collision-sensing, there is no need for error-correction/detecting coding, attackability estimation and synchronization. In this case, the proposed algorithms recover the collision-sensing algorithm in \cite{alatur2020multi}.

	\subsection{The centralized algorithm}\label{app:cent}
	For the completeness of this work, we first provide a regret analysis of the algorithm described in Section \ref{sec:arm_assign}, which is based on \cite{alatur2020multi}.
	\begin{theorem}
		The exploration regret of the centralized algorithm with $\eta=\sqrt{\log\tbinom{K}{M}/TKM}$ described in Section \ref{sec:arm_assign} has a regret upper bound of:
		\begin{equation}\label{eqn:c_regret}
		\mathbb{E}[R_c(T)] \leq 2M\sqrt{K\log(K)T}.
		\end{equation}
	\end{theorem}
	\begin{proof}
		First, we show that $\tilde{l}_A(t) = \sum_{k\in A}\tilde{l}_k(t)$ is an unbiased estimation of the true loss from meta-arm $A$ at time $t$, i.e., $l_A(t)=\sum_{k\in A}l_k(t)$. Denoting $P(t)=\{P_A(t)\}_{A\in\mathcal{K}}$, we first show that for any arm $k\in[K]$, $\tilde{l}_k(t)$ is an unbiased estimation of $l_k(t)$ as
		\begin{align*}
		&\mathbb{E}\left[\tilde{l}_k(t)|P(t)\right]\\
		&\overset{(i)}{=} M \frac{l_k(t)}{\sum_{A:k\in A\in\mathcal{K}}P_A(t)} \mathbb{P}(k=A_1(t))\\
		& = M \frac{l_k(t)}{\sum_{A:k\in A\in\mathcal{K}}P_A(t)} \mathbb{P}\left(k\in A(t))\mathbb{P}(k = A_1(t)|k\in A(t)\right)\\
		&\overset{(ii)}{ = } M \frac{l_k(t)}{\sum_{A:k\in A\in\mathcal{K}}P_A(t)} \sum_{A:k\in A\in\mathcal{K}} P_A(t)  \frac{1}{M} \\
		&= l_k(t),
		\end{align*}
		where equations (i) and (ii) are from the definitions of $\tilde{l}_k(t)$ and $P_A(t)$, respectively. From the law of total expectation, we can derive $\mathbb{E}[\tilde{l}_k(t)]=\mathbb{E}[\mathbb{E}[\tilde{l}_k(t)|P(t)]] = l_k(t)$. Finally, with $\tilde{l}_A(t) = \sum_{k\in A}\tilde{l}_k(t)$, by the linearity of expectation,  $\tilde{l}_A(t)$ is an unbiased estimation of $l_A(t)$.

		With the standard EXP3 regret guarantees, the centralized regret \cite{auer2002nonstochastic,alatur2020multi} is bounded as:
		\begin{equation*}
			\mathbb{E}\left[R_c(T)\right] \leq \eta\sum_{t=1}^T\sum_{A\in \mathcal{K}}\mathbb{E}\left[P_A(t) \mathbb{E}\left[(\tilde{l}_A(t))^2|P(t)\right]\right]+\frac{\log\tbinom{K}{M}}{\eta}.
		\end{equation*}
		The term $\mathbb{E}[(\tilde{l}_A(t))^2|P(t)]$ can be simplified as:
		\begin{align*}
		\mathbb{E}&\left[(\tilde{l}_A(t))^2|P(t)\right] \\
		=& \mathbb{E}\left[\left(\sum\nolimits_{k\in A}\tilde{l}_k(t)\right)^2|P(t)\right] \\
		=& \sum_{j,k \in A}\mathbb{E}\left[\tilde{l}_k(t) \tilde{l}_j(t)|P(t)\right] \\
		\overset{(i)}{=} &\sum_{k \in A }\mathbb{E}\left[(\tilde{l}_k(t))^2|P(t)\right]  \\
		=& \sum_{k \in A}\Big(\frac{M l_k(t)}{\sum_{B:k\in B\in\mathcal{K}}P_B(t)}\Big)^2\mathbb{P}(k=A_1(t))\\
		 = &\sum_{k \in A}\Big(\frac{M l_k(t)}{\sum_{B:k\in B\in\mathcal{K}}P_B(t)}\Big)^2\mathbb{P}(k\in A(t))\\
		&\cdot \mathbb{P}(k = A_1(t)|k\in A(t)) \\
		= &M\sum_{k \in A}\frac{(l_k(t))^2}{\sum_{B:k\in B\in\mathcal{K}}P_B(t)},
		\end{align*}
		where equation (i) is because $\tilde{l}_k(t)\neq 0$ holds for at most one arm. With this result, we get:
		\begin{align*}
			&\mathbb{E}\left[R_c(T)\right] \\
			&\leq  \eta\sum_{t=1}^T\sum_{A\in \mathcal{K}}\mathbb{E}\left[P_A(t)M\sum_{k \in A}\frac{(l_k(t))^2}{\sum_{B:k\in B\in\mathcal{K}}P_B(t)}\right]+\frac{\log\tbinom{K}{M}}{\eta}\\
			&= M\eta\sum_{t=1}^T\mathbb{E}\left[\sum_{A\in \mathcal{K}}P_A(t)\sum_{k \in A}\frac{(l_k(t))^2}{\sum_{B:k\in B\in\mathcal{K}}P(t)}\right]+\frac{\log\tbinom{K}{M}}{\eta}\\
			& = M\eta\sum_{t=1}^T\mathbb{E}\left[\sum_{k\in [K]}\frac{(l_k(t))^2\sum_{A:k\in A\in\mathcal{K}}P_A(t)}{\sum_{B:k\in B\in\mathcal{K}}P_B(t)}\right]+\frac{\log\tbinom{K}{M}}{\eta}\\
			& = M\eta\sum_{t=1}^T\mathbb{E}\left[\sum\nolimits_{k\in [K]}(l_k(t))^2\right]+\frac{\log\tbinom{K}{M}}{\eta}\\
			&\overset{(i)}{\leq} MKT\eta+\frac{\log\tbinom{K}{M}}{\eta}\\
			&\overset{(ii)}{=} 2\sqrt{MKT\log\tbinom{K}{M}} \\
			&\leq 2M\sqrt{KT\log(K)},
		\end{align*}
		where inequality (i) is the result of $l_k(t)\leq 1$, and equation (ii) is from $\eta=\sqrt{\log\tbinom{K}{M}/MKT}$.
	\end{proof}

	The following result establishes a regret bound for the blocked version of the centralized algorithm, which is important for the ensuing analysis in the decentralized setting.

	\begin{theorem}[\cite{alatur2020multi,arora2012online}]
		\label{thm:block}
		Let $\Pi$ be a bandit algorithm with an expected regret upper bound of $R(T)$. Then the blocked version of $\Pi$ with a block size $\tau$ has a regret upper bound of $\tau R(T/\tau)+\tau$.
	\end{theorem}
	In the multi-player case, the last term $\tau$, which represents the additional regret when $T$ is not divisible by $\tau$, is converted into $M\tau$.

	\subsection{$\alpha$-aware}\label{app:a_aware}
	The overall regret of $\alpha$-aware A2C2  can be decomposed as: $R_1(T)=R_1^{expl}(T)+R_1^{comm}(T)$, with $R_1^{expl}(T)$ and $R_1^{comm}(T)$ referring to the exploration and communication regret, respectively.

	The local attackability $W(T)$ satisfies $W(T)\leq O(T^{\alpha})$. Intuitively, the length of repetition code in $\alpha$-aware A2C2 is $h(T, \alpha+\epsilon)=\Theta(T^{\alpha+\epsilon})$. It is implied that $h(T,\alpha+\epsilon)=\omega(W(T))$, which means asymptotically no successful attack can possibly happen and communication phases are guaranteed to be successful.  With no communication phase in the centralized algorithm, the upper bound in Eqn.~\eqref{eqn:c_regret} can also serve as an upper bound for the exploration regret of $\alpha$-aware A2C2 after successful communications with Theorem \ref{thm:block}. Thus, the exploration regret $R_1^{expl}(T)$ can be bounded in the following lemma.
	\begin{lemma}\label{lem:a_aware_expl}
		The exploration regret of $\alpha$-aware A2C2  satisfies:
		\begin{align*}
		    \mathbb{E}\left[R_1^{expl}(T)\right]\leq O\left(M\sqrt{K\log(K)T\tau}\right).
		\end{align*}
	\end{lemma}

	\begin{proof}
		Since $h(T,\alpha+\epsilon)=\Theta(T^{\alpha+\epsilon})$, there exist $m_1>0, m_2>0$, and $T_0$ such that, $\forall T> T_0, m_1T^{\alpha+\epsilon}\leq h(T,\alpha+\epsilon)\leq m_2T^{\alpha+\epsilon}$. Furthermore, since $W(T)=O(T^{\alpha})$, $\exists n>0$ and $T_1$ such that $\forall T> T_1,  |W(T)|\leq nT^{\alpha}$. Thus, $\exists T^*=\max\{T_0,T_1, (n/m_1)^{\frac{1}{\epsilon}}\}$, $h(T,\alpha+\epsilon)\geq m_1T^{\alpha+\epsilon}> nT^{\alpha}\geq W(T)$, $\forall T> T^*$, which means communications are guaranteed to be successful for $T>T^*$. Thus, all explorations are collision-free. Denote $T_e=T-T_c$, where $T_e$ is the overall exploration time and $T_c$ is the overall communication time. With Theorem \ref{thm:block} and Eqn.~\eqref{eqn:c_regret}, we have that $\forall T>T^*$,
		\begin{align*}
		\mathbb{E}\left[R_1^{expl}(T)\right]&\leq \tau R_c(T_e/\tau)+M\tau\\
		&=2M\sqrt{K\log(K)T_e\tau}+M\tau\\
		&\leq 3M\sqrt{K\log(K)T\tau},
		\end{align*}
		and thus $\mathbb{E}[R_1^{expl}(T)]\leq O(M\sqrt{K\log(K)T\tau})$.
	\end{proof}
	There are at most $\left\lceil\frac{T}{\tau}\right\rceil$ rounds communication, while each round contains $(M-1)\lceil \log_2(K)\rceil h(T,\alpha+\epsilon)$ time slots. Thus, the communication regret can be bounded as follows.
	\begin{lemma}\label{lem:a_aware_comm}
		The communication regret of $\alpha$-aware A2C2 satisfies:
		\begin{align*}
		\mathbb{E}\left[R_1^{comm}(T)\right]&\leq M(M-1) \left\lceil \log_2(K)\right\rceil \left\lceil {T}/{\tau}\right\rceil h(T,\alpha+\epsilon)\\
		&\leq O\left(M^2\log(K)T^{1+\alpha+\epsilon}/\tau\right).
		\end{align*}
	\end{lemma}
	With Lemmas \ref{lem:a_aware_expl} and  \ref{lem:a_aware_comm}, Theorem \ref{thm:a_aware} can be proven via solving the optimization problem of
	\begin{equation*}\small
		\min_{\tau\in\mathbb{N}}\max\Big\{O\left(M\sqrt{K\log(K)T\tau}\right), O\Big(M^2\log(K)\frac{T^{1+\alpha+\epsilon}}{\tau}\Big)\Big\}
	\end{equation*}
	which leads to $\tau = \left \lceil M^{\frac{2}{3}}K^{-\frac{1}{3}}\log(K)^{\frac{1}{3}}T^{\frac{1+2\alpha+2\epsilon}{3}} \right \rceil$.

	\subsection{$\beta$-aware}\label{app:b_aware}
	The same error detection code in Section \ref{subsec:a_unaware} is also used in the $\beta$-aware A2C2 algorithm. However, it is not used for the update of estimations, but rather only to maintain an unbiased loss estimation. The $\beta$-aware A2C2 algorithm is presented in Algorithms \ref{alg:b_aware_leader} (leader) and \ref{alg:b_aware_follower} (follower).

	\begin{algorithm}[htb]
		\small
		\caption{$\beta$-aware A2C2: Leader}
		\label{alg:b_aware_leader}
		\begin{algorithmic}[1]
		\Require $M$, $K$, $T$
		\State \textbf{Initialize:} $\tau\gets \lceil K^{\frac{1}{3}}\log(K)^{-\frac{1}{3}}T^{\max\{\beta,\frac{1}{3}\}}\rceil; \eta\gets\sqrt{\log\tbinom{K}{M}\tau/MKT}; \nu\gets\max\left\{\frac{3\beta-1}{2},0\right\}$; $F\gets 0$
		\For{$p=1,2,...$}
			\State $\forall A\in\mathcal{K}, \tilde{L}_{A}(p)\gets \sum_{v=1}^{p-1}\sum_{k\in A}\tilde{l}_k(v)$
			\State $\forall A\in\mathcal{K}, P_A(p)\gets\frac{e^{-\eta \tilde{L}_A(p)}}{\sum_{J\in\mathcal{K}}e^{-\eta \tilde{L}_{J}(p)}}$ \Comment{Loss estimator}
			\State Choose $A(p)=\{A_1(p),...,A_M(p)\}$ with $P_A(p)$
			\State Randomly permute $A(p)$ into $\tilde{A}(p)$
			\Statex $\triangleright$ \textbf{Communication Phase:}
			\State $\forall m\in [M]$, $\text{msg}_m\gets \text{eEncoder}(\tilde{A}_m(p), k(T,\nu))$
			\State $\forall m\in [M]$, Send$\left(m, \text{msg}_m\right)$ \Comment{Send Assignment}
			\State $\text{msg}_F\gets \text{Receive}\left(k(T,\nu)\right)$
			\State $F \gets$ rDecoder$\left(\text{msg}_F, k(T,\nu)\right)$
			\Statex $\triangleright$ \textbf{Exploration Phase:}
			\State Stay on arm $\tilde{A}_1(p)$ for $\tau$ time steps\Comment{Exploration}
			\If{$F=0$}
			    \State Record cumulative loss $\hat{l}_{\tilde{A}_1(p)}(p)$
				\State $\forall k\in[K]$, set $\tilde{l}_k(p)\gets\frac{M}{\tau}\frac{\hat{l}_{\tilde{A}_1 (p)}(p) \mathds{1}\{\tilde{A}_1 (p)=k\}}{\sum_{A:k\in A\in \mathcal{K}}P_A(p)}$
			\Else \ $\forall k\in[K], \tilde{l}_k(p)\gets 0$
			\EndIf
		\EndFor
		\end{algorithmic}
	\end{algorithm}
	\begin{algorithm}[htb]
		\small
		\caption{$\beta$-aware A2C2: Follower}
		\label{alg:b_aware_follower}
		\begin{algorithmic}[1]
			\Require $M$, $K$, $T$, index $m$
			\State \textbf{Initialize:} $\tau\gets \lceil K^{\frac{1}{3}}\log(K)^{-\frac{1}{3}}T^{\max\{\beta,\frac{1}{3}\}}\rceil; \eta\gets\sqrt{\log\tbinom{K}{M}\tau/MKT}; \nu\gets\max\left\{\frac{3\beta-1}{2},0\right\}$; $F\gets 0$; $\mathcal{S}\gets\emptyset$
			\For{$p=1,2,...$}
				\Statex $\triangleright$ \textbf{Communication Phase:}
				\State $\text{msg}_m\gets \text{Receive}\left(k(T,\nu)\right)$
				\State $\mathcal{S}\gets$ eDecoder$\left(\text{msg}_m, k(T,\nu)\right)$\Comment{Receive Assignment}
				\State Randomly choose in $\mathcal{S}$ for $\tilde{A}_m(p)$
				\State  $F \gets \mathds{1}\left\{|\mathcal{S}|>1\right\}$
				\State $\text{msg}_F\gets \text{rEncoder}\left(k(T,\nu)\right)$
				\State Send$\left(1, \text{msg}_F\right)$ \Comment{Feedbak Communication Error}
				\Statex $\triangleright$ \textbf{Exploration Phase:}
				\State Stay on arm $\tilde{A}_m(p)$ for $\tau$ time steps \Comment{Exploration}
			\EndFor
		\end{algorithmic}
	\end{algorithm}

	The overall regret of $\beta$-aware A2C2 can be decomposed as $R_2(T)=R_2^{expl}(T)+R_2^{err}(T)+R_2^{comm}(T)$, where $R_2^{expl}(T)$ and $R_2^{err}(T)$ refer to the exploration regret after successful and failed communications, respectively, and $R_2^{comm}(T)$ characterizes the communication regret. With Theorem \ref{thm:block} and Eqn.~\eqref{eqn:c_regret}, $R_2^{expl}(T)$ can be upper bounded, as stated in the following lemma.
	\begin{lemma}\label{lem:b_aware_expl}
		The exploration regret of $\beta$-aware A2C2 after successful communication satisfies:
		$$
		\mathbb{E}\left[R_2^{expl}(T)\right] \leq  O\left(M\sqrt{K\log(K)T\tau}\right).
		$$
	\end{lemma}

	\begin{proof}
		Denote $T_s=T-T_c-T_f$, where $T_f$ is the overall exploration time after failed communications and $T_c$ is the overall communication time. Since $\forall t\in T_s$, all explorations are collision-free, and the same combination of Theorem \ref{thm:block} and Eqn.~\eqref{eqn:c_regret} leads to the following upper bound:
		\begin{align*}
		\mathbb{E}\left[R_2^{expl}(T)\right]&\leq \tau R_c(T_s/\tau)+M\tau\\
		&=2M\sqrt{K\log(K)T_s\tau}+M\tau\\
		&\leq 3M\sqrt{K\log(K)T\tau},
		\end{align*}
		which means $\mathbb{E}[R_2^{expl}(T)]\leq O(M\sqrt{K\log(K)T\tau})$.
	\end{proof}

	The overall number of loss ones on one arm, i.e., the global attackability, is $V(T)\leq O(T^{\beta})$. To succeed in attacking one communication phase, at least $k(T,\nu)=\Theta(T^{\nu})$ of loss ones are required. Thus, the overall number of successful attacks is no more than $\frac{MV(T)}{k(T,\nu)}\leq O(MT^{\beta-\nu})$ times	. The second term $R_2^{err}(T)$ in the overall regret can be bounded as follows.
	\begin{lemma}\label{lem:b_aware_err}
		The exploration regret caused by failed communications of $\beta$-aware A2C2 satisfies:
		\begin{align*}
		    \mathbb{E}\left[R_2^{err}(T)\right] \leq O\left(M^2T^{\beta-\nu}\tau\right).
		\end{align*}
	\end{lemma}

	\begin{proof}
	Since $k(T,\nu)=\Theta(T^{\nu})$, there exist $m_1>0, m_2>0$ and $T_0$ such that, $\forall T> T_0, m_1T^{\nu}\leq k(T,\nu)\leq m_2T^{\nu}$. Since $V(T)=O\left(T^{\beta}\right)$, $\exists n>0$ and  $T_1$ such that, $\forall T> T_1,  |V(T)|\leq nT^{\beta}$. Thus, $\exists T^*=\max\left\{T_0, T_1\right\}$,
	$\frac{V(T)}{k(T,\nu)}\leq \frac{n}{m_1}T^{\beta-\nu},$
	$\forall T>T^*$. With a loss of at most $M\tau$ caused by each successful attack, we have that $\forall T>T^*$,
	$$
	\mathbb{E}\left[R_2^{err}(T)\right]=M\frac{V(T)}{k(T,\nu)}\cdot M\tau\leq M^2\frac{n}{m_1}T^{\beta-\nu}\tau.
	$$
	Thus, $\mathbb{E}\left[R_2^{err}(T)\right]\leq O\left(M^2T^{\beta-\nu}\tau\right)$.
	\end{proof}

	With at most $\left\lceil {T}/{\tau}\right\rceil$ rounds of communications consisting of $M-1$ times of arm assignment and one time error report, the communication regret can be bounded as follows.
	\begin{lemma}\label{lem:b_aware_comm}
		The communication regret of $\beta$-aware A2C2 satisfies:
		\begin{align*}
		&\mathbb{E}\left[R_2^{comm}(T)\right]\\
		&\leq M(M-1)K \left\lceil T/\tau\right\rceil k(T,\nu)+M\left\lceil  T/\tau \right\rceil k(T,\nu)  \\
		&\leq O\left(M^2KT^{1+\nu}/\tau\right).
		\end{align*}
	\end{lemma}

	With Lemmas \ref{lem:b_aware_expl} to \ref{lem:b_aware_comm}, Theorem \ref{thm:b_aware} can be proven via solving the following optimization problem:
	\begin{align*}
		\min_{\tau\in\mathbb{N},\nu\geq 0}\max\Big\{&O\left(M\sqrt{K\log(K)T\tau}\right),\\
		&O\left(M^2T^{\beta-\nu}\tau\right), O\left(M^2KT^{1+\nu}/\tau\right)\Big\}
	\end{align*}
	which leads to $\tau = \lceil K^{\frac{1}{3}}\log(K)^{-\frac{1}{3}}T^{\max\left\{\beta,\frac{1}{3}\right\}}\rceil$ and $\nu=\max\{\frac{3\beta-1}{2},0\}$.

	\subsection{$\alpha$-unaware}
	\label{app:a_unaware}
	Since multiple estimations are required before successful communications can be guaranteed, $R_3^{expl}(T)$ consists of not only the exploration regret $R_3^{expl}(T|\zeta=1)$ after successful communication phases, but also explorations regret $R_3^{expl}(T|\zeta=0)$ after failed  communications, which is different from $R_1^{expl}(T)$ in $\alpha$-aware A2C2. The estimation $\alpha'= j\epsilon$ is denoted as $\alpha_j$ for an integer $j$ for simplicity. The overall exploration time steps when $\alpha_j$ is used as the $\alpha$-estimation is denoted as $T_{e,j}$, and the corresponding $\tau$ and $\xi$ with $\alpha_j$ are denoted as $\tau_j$ and $\xi_j$, respectively.

	With the update, $\exists v>0$, $\delta\in[0,\epsilon)$, $\alpha\leq \alpha+\delta\leq  \alpha_v=v\epsilon \leq \alpha+\epsilon$. Since $h(T,\alpha_v)=\Theta(T^{v\epsilon})$, $\exists \kappa_1>0, \kappa_2>0$ and $T_0$ such that, $\forall T> T_0$, $\kappa_1T^{v\epsilon}\leq h(T,\alpha_v)\leq \kappa_2T^{v\epsilon}$. Since $W(T)=O(T^{\alpha})$, $\exists n>0$ and $T_1$ such that, $\forall T> T_1,  |W(T)|\leq nT^{\alpha}$. It implies that $W(T)\leq nT^{\alpha}<\kappa_1T^{\alpha+\delta}\leq h(T,\alpha_v)$,
	$\forall T>T^*=\max\{T_0, T_1, (n/\kappa_1)^{\frac{1}{\delta}}\}$, which means $\forall T>T^*$, the repetition code can overpower the local attacks, and thus successful communications are guaranteed with $\alpha_v$. For a given adversary sequence, $\forall T>T^*$, the update is assumed to stop at $\alpha_w=w\epsilon$ with $w\leq v$.

	\mypara{Proof of Lemma \ref{lem:a_unaware_expl}.}
	Denote $T_{s,j}$ and $T_s = \sum_{j=0}^w T_{s,j}\leq \sum_{j=0}^w T_{e,j}\leq T$ as the length of exploration with successful preceding communications under estimation $\alpha_j$ and in total, respectively. The exploration under each estimation $\alpha_j$ with successful preceding communications can be viewed alone as a bandit game with horizon $T_{s,j}$ and block $\tau_j$. Thus, by applying Theorem \ref{thm:block} to Eqn.~\eqref{eqn:c_regret},  $\forall T>T^*$, the first term in the exploration regret can be bounded as
		\begin{align*}
		&\mathbb{E}\left [R_3^{expl}(T|\zeta=1) \right ]\\
		&= \sum_{j=0}^{w}\mathbb{E}\left[R_3^{expl}(T|\zeta=1, \alpha'=\alpha_j) \right ]\\
		&\leq \sum_{j=0}^w  (2M\sqrt{K\log(K)T_{s,j} \tau_j}+M\tau_j )\\
		&\overset{(i)}{\leq} 3(w+1)M\sqrt{K\log(K)T_{s} \tau_w}\\
		&\leq 3\sqrt{2}(v+1)M^{\frac{4}{3}}K^{\frac{1}{3}}\log(K)^{\frac{1}{3}}T^{\frac{5+\alpha_v}{6}},
		\end{align*}
		where inequality (i) is from that $T_{s,j}\leq T_s$ and $\tau_j$ is monotonically increasing with $j$. Thus, $\mathbb{E}[R_3^{expl}(T)|\zeta=1]\leq O( M^{\frac{4}{3}}K^{\frac{1}{3}}\log(K)^{\frac{1}{3}}T^{\frac{5+\alpha+\epsilon}{6}})$.

		The second term is caused by the potential collisions due to failed communications with an underestimated attackability, and $\forall T>T^*$, for the adversary given above, it takes $w$ trials to eliminate communication errors, and the trial with estimation $\alpha_j$ leads to a regret of $M\tau_j$. Thus, the second term can be bounded as:
		\begin{align*}
		&\mathbb{E}\left [R_3^{expl}(T|\zeta=0) \right]\\ &=\sum_{j=0}^{w-1}\mathbb{E}\left [R_3^{expl}(T|\zeta=0, \alpha'=\alpha_j)  \right ]\\
		&\leq \sum_{j=0}^{w-1}M\tau_j\\
		&=\sum_{j=0}^{w-1}M\left\lceil M^{\frac{2}{3}}K^{-\frac{1}{3}}\log(K)^{-\frac{1}{3}}T^{\frac{2+\alpha_j}{3}}\right\rceil\\
		&\leq  2M^{\frac{5}{3}}K^{-\frac{1}{3}}\log(K)^{-\frac{1}{3}}T^{\frac{2+\alpha_v}{3}}.
		\end{align*}
		Since $\alpha_v\leq 1$, the first term dominates the second term as $\frac{2+\alpha_v}{3}\leq \frac{5+\alpha_v}{6}$. Thus, the overall exploration regret can be bounded as $\mathbb{E}[R_3^{expl}(T)]\leq O(M^{\frac{4}{3}}K^{\frac{1}{3}}\log(K)^{\frac{1}{3}}T^{\frac{5+\alpha+\epsilon}{6}})$.

	\mypara{Proof of Lemma \ref{lem:a_unaware_comm}.}
		The communication regret consists of that for arm assignment and synchronization under different estimations.  The arm assignment in each round consists of $(M-1)Kh(T,\alpha_j)$ time slots while the synchronization has an average length of $Mh(T,\alpha_j)\frac{\left\lceil T^{\xi_j}\right\rceil}{2}$. For the given adversary in the proof above, $\forall T>T^*$, the overall communication regret can be bounded as:
		\begin{align*}
		&\mathbb{E}\left[R_3^{comm}(T)\right]\\
		&\leq\sum_{j=0}^{w}\Big\lceil \frac{T_{e,j}}{\tau_j}\Big\rceil \Big(M(M-1)Kh(T,\alpha_j)+M^2h(T,\alpha_j)\frac{\lceil T^{\xi_j}\rceil}{2}\Big)
		\end{align*}
		Considering $h(T,\alpha_j)=\Theta(T^{\alpha_j})$, $\forall j=0,1,...,v$, $\exists f_{j,1}>0, f_{j,2}>0$ and $T_{j,0}$ such that, $\forall T>T_{j,0}$, $f_{j,1}T^{\alpha_j}\leq h(T,\alpha_j)\leq f_{j,2}T^{\alpha_j}$. Thus, $\forall j=0,1,...,v$, $\exists f_{1}\leq \min_j\{f_{j,1}\}, \exists f_{2}\geq \max_j\{f_{j,2}\}$ and $T_{x}>\max_j\{T_{j,0}\}$ such that $\forall T>T_x$, $f_{1}T^{\alpha_j}\leq h(T,\alpha_j)\leq f_{2}T^{\alpha_j}$. Based on this, $\forall T>\max\{T^*,  T_x\}$, the communication regret can be bounded as:
		\begin{align*}
	      \mathbb{E}&[R_3^{comm}(T)] \\
	      \leq &  2\sum_{j=0}^{w}\frac{T_{e,j}}{\tau_j} \left(M^2Kf_2T^{\alpha_j}+M^2f_2T^{\alpha_j}T^{\xi_j}\right)\\
		= &2\sum_{j=0}^{w}T_{e,j} M^{\frac{4}{3}}K^{\frac{4}{3}}\log(K)^{\frac{1}{3}}f_2T^{\frac{\alpha_j-1}{3}}\\
		&+2\sum_{j=0}^{w}T_{e,j}M^{\frac{4}{3}}K^{\frac{1}{3}}\log(K)^{\frac{1}{3}}f_2T^{\frac{\alpha_j-1}{6}}\\
		\leq & 2M^{\frac{4}{3}}K^{\frac{4}{3}}\log(K)^{\frac{1}{3}}f_2T^{\frac{2+\alpha_v}{3}}\\
		&+2M^{\frac{4}{3}}K^{\frac{1}{3}}\log(K)^{\frac{1}{3}}f_2T^{\frac{5+\alpha_v}{6}}.
		\end{align*}
		With $\alpha_v< 1$, we have $\mathbb{E}[R_3^{comm}(T)] \leq O  (M^{\frac{4}{3}}K^{\frac{1}{3}}\log(K)^{\frac{1}{3}}T^{\frac{5+\alpha+\epsilon}{6}})$.

%
	\mypara{Proof of Lemma \ref{lem:a_unaware_sync}.}
	    Before completing the estimation of $\alpha'$, each communication for synchronization has a failure probability of $\frac{1}{\lceil T ^{\xi}\rceil}$, which in the worst case has a linear regret $MT$. With a union bound for all communication phases with estimation less than $\alpha_w$, we have:
		\begin{align*}
		\mathbb{E}[R_3^{sync}(T)]&\leq\sum_{j=0}^{w-1}\left\lceil\frac{T_{e,j}}{\tau_j}\right\rceil \frac{1}{\lceil T^{\xi_j}\rceil}MT\\
		&\leq 2 \sum_{j=0}^{w-1}\frac{T_{e,j}}{\tau_j} \frac{1}{T^{\xi_j}}MT\\
		&\leq 2M^{\frac{1}{3}}K^{\frac{1}{3}}\log(K)^{\frac{1}{3}}T^{\frac{5+\alpha+\epsilon}{6}}.
		\end{align*}
	Thus, Lemma \ref{lem:a_unaware_sync} can be obtained as
 	$\mathbb{E}[R_3^{sync}(T)]\leq O(M^{\frac{1}{3}}K^{\frac{1}{3}}\log(K)^{\frac{1}{3}}T^{\frac{5+\alpha+\epsilon}{6}}).$

	\subsection{$\beta$-unaware}
	\label{app:b_unaware}
    The $\beta$-unaware algorithm for the leader and followers are presented in Algorithms \ref{alg:b_unaware_leader} and \ref{alg:b_unaware_follower}, respectively.  The following proofs focus on $\beta\geq \frac{1}{4}$, and the case of $\beta\leq \frac{1}{4}$ can be obtained as a special case where the estimation $\beta'$ is kept as $\frac{1}{4}$. Under estimation $\beta'=\frac{1}{4}+j\epsilon$, denoted as $\beta_j$ for simplicity, similar notations of $T_{e,j}$, $\tau_j$, $\xi_j$ and $\nu_j$ are applied referring to the overall time that the current estimation holds and the corresponding parameters.

    With the update, $\exists v>0, \delta\in(0,\epsilon], \beta+\epsilon \geq\beta'=\beta_v> \beta+\delta\geq \beta$. Since $k(T,\beta_v)=\Theta(T^{\frac{1}{4}+v\epsilon})$, $\exists \kappa_1>0$, $\kappa_2>0$, and $T_0$ such that, $\forall T> T_0$, $\kappa_1T^{\beta+\delta}\leq \kappa_1T^{\frac{1}{4}+v\epsilon}\leq k(T,\beta_v)\leq \kappa_2T^{\frac{1}{4}+v\epsilon}\leq \kappa_2T^{\beta+\epsilon}$. With $V(T)=O(T^{\beta})$, $\exists n>0$, $T_1$ so that $\forall T> T_1$,  $|V(T)|\leq nT^{\beta}$. We then have $V(T)\leq nT^{\beta}<\kappa_1 T^{\beta+\delta}\leq k(T,\beta_v)$,
	$\forall T>T^*=\max\{T_0, T_1, (n/\kappa_1)^{\frac{1}{\delta}}\}$, which means updating stops with $\beta_v$. For a given adversary, $\forall T>T^*$, the estimation is assumed to be completed as $\beta'=\beta_w=\frac{1}{4}+w\epsilon$, where $w\leq v$.

    \begin{algorithm}[htb]
    	\small
    	\caption{$\beta$-unaware A2C2: Leader}
    	\label{alg:b_unaware_leader}
    	\begin{algorithmic}[1]
    		\Require{$M$, $K$, $T$}
    		\State \textbf{Initialize:} $\beta'\gets \frac{1}{4}$; communication flag: $F_1, F_2\gets 0$; round counter: $R\gets 0$; collision counter $C\gets 0$
    		\For{$p=1,2,...$}
    		\State $\tau\gets \lceil K^{-\frac{1}{3}}\log(K)^{-\frac{1}{3}}T^{\frac{1+2\beta'}{3}}\rceil$ \State $\eta\gets\sqrt{\log\tbinom{K}{M}\tau/MKT}$;  $\xi\gets\frac{1+2\beta'}{3}$; $\nu'\gets\frac{4\beta'-1}{3}$
    		\State $F_1\gets 0$; $R\gets R + 1$;
    		\State $\forall A\in\mathcal{K}, \tilde{L}_{A}(p)\gets \sum_{v=1}^{p-1}\sum_{k\in A}\tilde{l}_k(v)$
    		\State $\forall A\in\mathcal{K}$, $P_A(p)\gets\frac{e^{-\eta \tilde{L}_A(p)}}{\sum_{J\in\mathcal{K}}e^{-\eta \tilde{L}_{J}(p)}}$
    		\State Choose $A(p)=\{A_1(p),...,A_M(p)\}$ with $P_A(p)$
    		\State Randomly permute $A(p)$ into $\tilde{A}(p)$
    		\Statex $\triangleright$ \textbf{Communication Phase:}
    		\State $\forall m\in [M]$, $\text{msg}_m\gets \text{eEncoder}(\tilde{A}_m(p), k(T,\nu'))$
    		\State $\forall m\in [M]$, Send$(m, \text{msg}_m)$\Comment{Send Assignment}
    		\State $\text{msg}_{F_1}\gets \text{Receive}(k(T,\nu'))$
    		\State $F_1 \gets$ rDecoder$(\text{msg}_{F_1},k(T,\nu'))$; \Comment{Uplink}
    		\State $C\gets C+F_1k(T,\nu')$\Comment{Count Attack}
    		\If{$R\geq \left\lceil\frac{ T^{\beta'} }{k(T,\nu')}\right\rceil$}
    		\State $F_2\gets \mathds{1}\{C\geq \lceil T^{\beta'}\rceil\}$\Comment{Update Point}
    		\For{$q=1,2,..., N(\xi)$}
    		\State $\text{msg}_{F_2}\gets \text{rEncoder}\left(F_2, k(T,\beta')\right)$
    		\State $\forall m\in [M]$, Send$(m, \text{msg}_{F_2})$ \Comment{Downlink}
    		\State $\text{msg}_{F_2}\gets \text{Receive}(k(T,\beta'))$
    		\State $F_2 \gets$ rDecoder$(\text{msg}_{F_2}, k(T,\beta'))$ \Comment{Uplink}
    		\EndFor
    		\State $R\gets 0$; $\beta'\gets\beta'+F_2\epsilon$\Comment{Update}
    		\EndIf
    		\Statex $\triangleright$ \textbf{Exploration Phase:}
    		\State Stay on arm $\tilde{A}_1(p)$ for $\tau$ time steps\;
    		\If{$F_1=0$}
    		\State Record cumulative loss $\hat{l}_{\tilde{A}_1(p)}(p)$
    		\State $\forall k\in[K]$, set $\tilde{l}_k(p)\gets\frac{M}{\tau}\frac{\hat{l}_{\tilde{A}_1 (p)}(p)\mathds{1}\{\tilde{A}_1 (p)=j\}}{\sum_{A:k\in A\in \mathcal{K}}P_A(p)}$
    		\Else \ $\forall k\in[K]$, $\tilde{l}_k(p)\gets 0$
    		\EndIf
    		\EndFor
    	\end{algorithmic}
    \end{algorithm}

    \begin{algorithm}[htb]
    	\small
    	\caption{$\beta$-unaware A2C2: Follower}
    	\label{alg:b_unaware_follower}
    	\begin{algorithmic}[1]
    		\Require $M$, $K$, $T$, index $m$
    		\State \textbf{Initialize:} $\beta'\gets \frac{1}{4}$; communication flag: $F_1, F_2\gets 0$; round counter: $R\gets 0$; collision counter $C\gets 0$
    		\For{$p=1,2,...$}
    		\State $\tau\gets \lceil K^{-\frac{1}{3}}\log(K)^{-\frac{1}{3}}T^{\frac{1+2\beta'}{3}}\rceil$
    		\State $\eta\gets\sqrt{\log\tbinom{K}{M}\tau/MKT}$; $\xi\gets\frac{2-2\beta'}{3}$; $\nu'\gets\frac{4\beta'-1}{3}$
    		\State $F_1,F_2\gets 0$; $R\gets R+1$
    		\Statex $\triangleright$ \textbf{Communication Phase:}
    		\State $\text{msg}_m\gets \text{Receive}(k(T,\nu'))$
    		\State $\mathcal{S}\gets$ eDecoder$(\text{msg}_m, k(T,\nu'))$\Comment{Receive Assignment}
    		\State Randomly choose in $\mathcal{S}$ for $\tilde{A}_m(p)$
    		\State $F_1 \gets \mathds{1}\{|\mathcal{S}|>1\}$ \Comment{Communication Error}
    		\State $C \gets C+(|S|-1)k(T,\nu')$\Comment{Count Attack}
    		\State $\text{msg}_{F_1}\gets \text{rEncoder}(F_1, k(T,\nu'))$
    		\State Send$(1, \text{msg}_{F_1})$\Comment{Uplink}
    		\If{$R\geq \left\lceil\frac{ T^{\beta'} }{k(T,\nu')}\right\rceil$}
    		\State $F_2\gets \mathds{1}\{C\geq \lceil T^{\beta'}\rceil\}$
    		\For{$q=1,2,..., N(\xi)$}\Comment{Update Point}
    		\State $\text{msg}_{F_2}\gets \text{Receive}( k(T,\beta'))$\Comment{Downlink}
    		\State $F_2\gets\max\{F_2,\text{rDecoder}( \text{msg}_{F_2},k(T,\beta'))\}$
    		\State $\text{msg}_{F_2}\gets \text{rEncoder}(F_2, k(T,\beta'))$
    		\State Send$(1, \text{msg}_{F_2})$\Comment{Uplink}
    		\EndFor
    		\State $R\gets 0$; $\beta'\gets\beta'+F_2\epsilon$;  $F_2\gets 0$\Comment{Update}
    		\EndIf
    		\Statex $\triangleright$ \textbf{Exploration Phase:}
    		\State Stay on arm $\tilde{A}_m(p)$ for $\tau$ time steps
    		\EndFor
    	\end{algorithmic}
    \end{algorithm}

    \mypara{Proof of Lemma \ref{lem:b_unaware_expl}.}
	 Denote $T_s = \sum_{j=0}^w T_{s,j}\leq T$ as the overall length of the exploration phases after successful communications. The term $R_4^{expl}(T)$ now consists of only explorations after successful communications, which is again bounded by Theorem \ref{thm:block} and Eqn.~\eqref{eqn:c_regret} as:
		\begin{align*}
		\mathbb{E}\left [R_4^{expl}(T) \right ] &= \sum_{j=0}^w \mathbb{E}\left [R_4^{expl}(T|\beta'=\beta_j) \right] \\
		&\leq \sum_{j=0}^w \big (2M\sqrt{K\log(K)T_{s,j}\tau_j}+M\tau_j \big)\\
		&\leq 3(v+1)M\sqrt{K\log(K)T_{s}\tau_v}\\
		&\leq 3(v+1)MK^{\frac{1}{3}}\log(K)^{\frac{1}{3}}T^{\frac{2+\beta_v}{3}}.
		\end{align*}
		Thus, we have $\mathbb{E}[R_4^{expl}(T)]\leq O( MK^{\frac{1}{3}}\log(K)^{\frac{1}{3}}T^{\frac{2+\beta+\epsilon}{3}} )$.

	\mypara{Proof of Lemma \ref{lem:b_unaware_err}.}
		With the estimated $\beta_j$, the adversary at most attacks $D_j=2 \left \lceil\frac{ T^{\beta_j} }{k(T,\nu_j)} \right \rceil$ rounds on each arm before updating, and each successful attack leads to a loss of at most $M\tau_j$. With $k(T,\nu_j)=\Theta(T^{\nu_j})=\Theta(T^{\frac{4\beta_j-1}{3}})$, $\forall j=0,1,...,w$, $\exists h_{j,1}>0,\exists h_{j,2}>0$, $\exists T_{m,y}$, $\forall T>T_{j,0}$, $h_{j,1}T^{\frac{4\beta_j-1}{3}}\leq k(T,\nu_j)\leq h_{j,2}T^{\frac{4\beta_j-1}{3}}$. Thus, $\forall j=0,1,...,v$, $\exists h_{1}\leq \min_j\{h_{j,1}\}$, $\exists h_{2}\geq \max_j\{h_{j,2}\}$, $\exists T_{y}\geq\max_j\{T_{j,0}\}$, $\forall T>T_y$, $h_{1}T^{\frac{4\beta_j-1}{3}}\leq k(T,\nu_j)\leq h_{2}T^{\frac{4\beta_j-1}{3}}$. It then follows that
		\begin{align*}
		\mathbb{E}\left[R_4^{err}(T) \right] &\leq \sum_{j=0}^w MD_j M\tau_j \\
		&\leq \sum_{j=0}^w \frac{4T^{\beta_j}}{k(T,\nu_j)} M^2\tau_j\\
		&\leq \sum_{j=0}^w\frac{4}{h_1} M^2K^{\frac{1}{3}}\log(K)^{-\frac{1}{3}}T^{\frac{2+\beta_j}{3}} \\
		&\leq (v+1)\frac{4}{h_1}  M^2K^{\frac{1}{3}}\log(K)^{-\frac{1}{3}}T^{\frac{3+\beta_v}{3}},
		\end{align*}
		$\forall T>\max\{T^*,T_y\}$, which means $\mathbb{E}[R_4^{err}(T)] \leq O (M^{2}K^{\frac{1}{3}}\log(K)^{-\frac{1}{3}}T^{\frac{2+\beta+\epsilon}{3}})$.

	\mypara{Proof of Lemma \ref{lem:b_unaware_comm}.}
		Communications consist of three parts: those for arm assignment, communication error report, and synchronization. With the estimated $\beta_j$, the arm assignment and communication error report happen at most $\lceil T_{e,j}/\tau_j\rceil$ rounds while each round lasts $(M-1)Kk(T,\nu_j)+k(T,\nu_j)$ time slots. On the other hand, under the same $\beta_j$, the synchronization happens at most $\left \lceil\frac{T_{e,j}}{\tau_j} \right \rceil\frac{k(T,\nu_j)}{ T^{\beta_j} }$ rounds while each round lasts $Mk(T,\nu_j)\frac{ \left \lceil T^{\xi_j} \right \rceil}{2}$ time slots on average. Thus, it can be get that
		\begin{align*}
		\mathbb{E}&[R_4^{comm}(T)]\\
		=&\sum_{j=0}^{w}\left\lceil \frac{T_{e,j}}{\tau_j}\right\rceil (M-1)\Big[MK k(T,\nu_j)+k(T,\nu_j)\Big]\\
		&+\sum_{j=0}^{w}\left\lceil\frac{T_{e,j}}{\tau_j}\right\rceil\frac{M^2k^2(T,\nu_j)}{ T^{\beta_j} }\frac{ \left \lceil T^{\xi_j} \right \rceil}{2}\\
	    \leq &2\sum_{j=0}^wh_2M^2K^{\frac{4}{3}}\log(K)^{\frac{1}{3}}T_{e,j}T^{\frac{2\beta_j-2}{3}}\\
		&+2\sum_{j=0}^w(h_2)^2M^2K^{\frac{1}{3}}\log(K)^{\frac{1}{3}}T_{e,j}T^{\frac{\beta_j-1}{3}}\\
		 \leq & 2h_2M^2K^{\frac{4}{3}}\log(K)^{\frac{1}{3}}T^{\frac{1+2\beta_v}{3}}\\
		&+2(h_2)^2M^2K^{\frac{1}{3}}\log(K)^{\frac{1}{3}}T^{\frac{2+\beta_v}{3}}.
		\end{align*}
		$\forall T> T^*$. Since $\beta_v\leq 1$, $\frac{2+\beta_v}{3}\geq \frac{1+2\beta_v}{3}$, $\mathbb{E}[R_4^{comm}(T)]\leq O (M^2K^{\frac{1}{3}}\log(K)^{\frac{1}{3}}T^{\frac{2+\beta+\epsilon}{3}}).$

	\mypara{Proof of Lemma \ref{lem:b_unaware_sync}.}
		With the estimated $\beta_j$, the synchronization happens every $\left \lceil\frac{ T^{\beta_j} }{k(T,\nu_j)} \right \rceil$ iterations of exploration and communication, and each synchronization has a failure probability $\frac{1}{\left \lceil T^{\xi_j}\right \rceil}$, which leads to a worst-case linear regret of $MT$ in the following time slots:
		\begin{align*}
		\mathbb{E}[R_4^{sync}(T)]&\leq \sum_{j=0}^w \left\lceil\frac{T_{e,j}}{\tau_j}\right\rceil \frac{k(T,\nu_j)}{ T^{\beta_j} }\frac{1}{ \left \lceil T^{\xi_j} \right \rceil} MT \\
		&\leq 2h_2\sum_{j=0}^w \frac{T_{e,j}}{\tau_j}MT^{1+\nu_j-\beta_j-\xi_j}\\
		&\leq 2h_2\sum_{j=0}^w MK^{\frac{1}{3}}\log(K)^{\frac{1}{3}}T_{e,j}T^{\frac{\beta_j-1}{3}}  \\
		&\leq 2h_2MK^{\frac{1}{3}}\log(K)^{\frac{1}{3}}T^{\frac{2+\beta_w}{3}}\\
		&\leq 2h_2MK^{\frac{1}{3}}\log(K)^{\frac{1}{3}}T^{\frac{2+\beta+\epsilon}{3}}.
		\end{align*}
		Thus, we have $\mathbb{E}[R_4^{sync}(T)]\leq O(MK^{\frac{1}{3}}\log(K)^{\frac{1}{3}}T^{\frac{2+\beta+\epsilon}{3}})$.

\subsection{Additional numerical illustrations}
\label{apdx:sim}

In addition to Figs.~\ref{fig:multiplayer} and \ref{fig:attack} in Section~\ref{sec:intro}, a more detailed numerical illustration is given here to compare A2C2 with the algorithm in \cite{bubeck2020non}.
	In Fig.~\ref{fig:heatmap_a}, the regret differences between $\alpha$-unaware and \cite{bubeck2020non} with different number of participating players (i.e., $M$) and attackabilties (i.e., $\alpha$)  are illustrated. We see that the advantage of $\alpha$-unaware A2C2 is more pronounced in the region of large $M$ and small $\alpha$. A similar observation can be made from Fig.~\ref{fig:heatmap_b}, which shows the regret differences between $\beta$-unaware A2C2 and \cite{bubeck2020non}.

	\begin{figure}[htb]
    		\centering
		\includegraphics[width=0.9\linewidth]{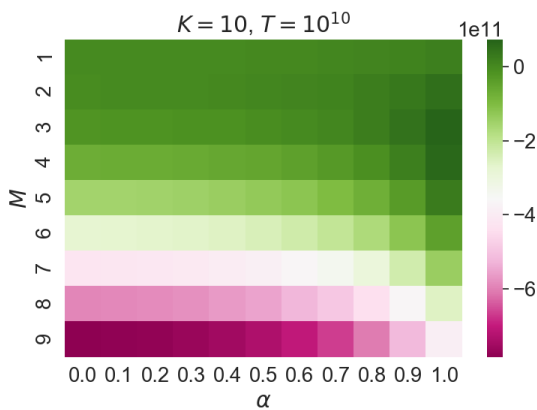}
		\caption{Regret differences between $\alpha$-unaware A2C2 and \cite{bubeck2020non}.}
		\label{fig:heatmap_a}
	\end{figure}
	\begin{figure}[htb]
		\centering
		\includegraphics[width=0.9\linewidth]{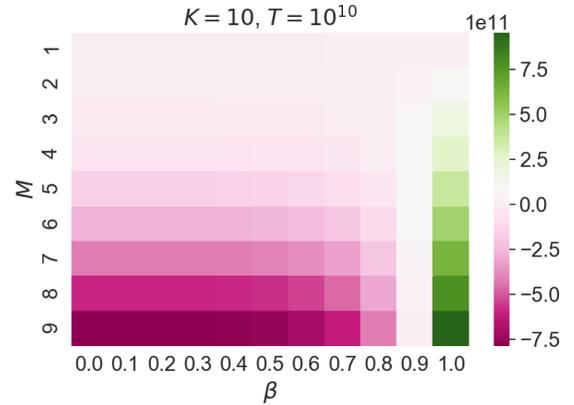}
		\caption{Regret differences between $\beta$-unaware A2C2 and \cite{bubeck2020non}.}
		\label{fig:heatmap_b}
	\end{figure}

\bibliographystyle{IEEEtran}
\bibliography{ref_final}

\end{document}